\documentclass{amsart} 
\pdfoutput=1

\usepackage{amsmath,color,graphicx, amssymb, amsthm, mathtools, mathrsfs, hyperref,caption,subcaption} 
\usepackage[section]{placeins}
\newtheorem{thm}{Theorem}[section]
\newtheorem*{thm*}{Theorem}

\newtheorem{prop}[thm]{Proposition}

\newtheorem{lem}[thm]{Lemma}

\theoremstyle{definition}
\newtheorem{defi}{Definition}[section]

\theoremstyle{remark}
\newtheorem*{rmk}{Remark}

\theoremstyle{remark}

\numberwithin{equation}{section}
\usepackage{url}

\renewcommand\Re{\operatorname{Re}}
\renewcommand\Im{\operatorname{Im}}
\newcommand\mult{\operatorname{mult}}
\newcommand\diag{\operatorname{diag}}

\allowdisplaybreaks
\begin{document}
\title[Dirac cones for a honeycomb point scatterer]{Dirac cones for point scatterers on a honeycomb lattice}
\author{Minjae Lee}
\email{lee.minjae@math.berkeley.edu}
\address{Department of Mathematics, University of California, Berkeley}
\date{\today}

\begin{abstract}
We investigate the spectrum and the dispersion relation of the Schrödinger operator with point scatterers on a triangular lattice and a honeycomb lattice. We prove that the low-level dispersion bands have conic singularities near Dirac points, which are the vertices of the first Brillouin Zone. The existence of such conic dispersion bands plays an important role in various electronic properties of honeycomb-structured materials such as graphene. We then prove that for a honeycomb lattice, the spectra generated by higher-level dispersion relations are all connected so the complete spectrum consists of at most three bands. Numerical simulations for dispersion bands with various parameters are also presented.
\end{abstract}
\maketitle

\section{Introduction}
In this paper we investigate the spectral properties of the Schrödinger operator with periodic point scatterers on the triangular lattice \(\Lambda\subset \mathbb{R}^2\) and the honeycomb lattice \(H \subset \mathbb{R}^2\) (See Figure~\ref{trihex}). The notion of point scatterers started with the Kronig-Penny model \cite{kronigpenney} which describes the dispersion relation, or energy-momentum relation, and the band structure of an electron on a 1-dimensional solid crystal. This idea has been generalized to infinitely many point scatterers on a periodic structure in \(\mathbb{R}^d, d\le 3\) \cite{solvable, Ka1, Ka2} using Krein's theory of self-adjoint extensions. We can borrow the notion of dispersion bands from the Floquet-Bloch theory \cite{GKT1, GKT2, solidstate, KF} for a periodic potential to analyze the periodic point scatterers as well. See Part III of \cite{solvable} for the formulation of periodic point scatterers in detail. 

Honeycomb lattices and regular periodic potentials with such symmetry structure have drawn considerable interest in the physics community due to the groundbreaking fabrication technique for graphene \cite{nobelprize}, a two-dimensional single layer of carbon atoms arranged in the honeycomb lattice structure. In 1947, using the tight-binding model, P.R.Wallace \cite{wallace} found that the dispersion relation of graphene have conic singularities at the six corners of the First Brillouin Zone. These points and singularities are called Dirac points and Dirac cones since the wave packet whose momentum components concentrated near those points behaves as a solution of two-dimensional Dirac wave equation, which describes the evolution of massless relativistic fermions. In addition, P. Kuchment and O. Post \cite{KP} in 2007 generalized the tight-binding model as a quantum graph with potential on the edges of the honeycomb lattice which also presents the conic singularities of the dispersion relation.


Fefferman and Weinstein \cite{honeycomb, honeycomb2} recently showed that generic smooth and real-valued honeycomb lattice potentials have conical singularities (Dirac points) at the vertices of the Brillouine zone. Note that this result does not depend on the magnitude of the potential. See also \cite{FWreview}.  We extend these ideas to the honeycomb lattice point scatterers, which can be formally considered as a singular potential on \(\mathbb{R}^2\) concentrated on a honeycomb lattice, and observe a similar type of conic singularities regardless of the strength of the point scatterers as follows:

\begin{thm*} (Theorem~\ref{conicthm}, Theorem~\ref{band}) Let \(\mathcal{H}\) be any self-adjoint operator on \(L^2(\mathbb{R}^2)\) acting as the Laplacian away from the honeycomb lattice \(H\) of \eqref{H}. Suppose that \(\mathcal{H}\) satisfies the honeycomb symmetry condition as in Proposition~\ref{honeycombsymm}. Then there exists the Dirac cone on the dispersion relation of \(\mathcal{H}\).  In addition, the spectrum of \(\mathcal{H}\) consists of at most three disjoint intervals on the real line.
\end{thm*}

In order to prove this theorem, we first investigate periodic point scatterers on a triangular lattice, which are similar but simpler compared to those on a honeycomb lattice. Then the spectral properties of point scatterers on a honeycomb lattice follow by modifying those results. More precisely, for both lattices, the spectral properties in terms of dispersion bands are studied locally near Dirac points and globally over the Brillouin zone. In Section~\ref{sec:triangular} for triangular lattice point scatterers, Theorem~\ref{prop:1detaconic} shows that the second and third dispersion bands form a pair of conic singularities at Dirac points. Then we introduce Proposition~\ref{prop:floq1spec} proving that the corresponding band spectrum has at most one gap and it occurs between the first and second dispersion bands. Similarly, in Section~\ref{sec:honeycomb} for honeycomb lattice point scatterers, Theorem~\ref{conicthm} shows that the first, second, fourth and fifth dispersion bands form two pairs of conic singularities at Dirac points. Furthermore, Theorem~\ref{band} shows that the band spectrum has at most two gaps and they occur between the second, third, and fourth dispersion bands. These results are illustrated numerically by various figures and movies. See Figure~\ref{floqaloc} ,\ref{floqaglobal}, \ref{floqspec} and Supplemental Materials for the conic dispersion bands and the corresponding spectra of periodic point scatterers. In Appendix~\ref{sec:introtoptsctr}, we briefly introduce the Floquet-Bloch theory and the notion of periodic point scatterers.

\section{Preliminaries}
\subsection{Lattice structures on a 2-dimensional space} \label{2lat}
We introduce three kinds of lattice structures: a triangular lattice, a honeycomb lattice, and the dual lattice of them. 
Let \( \mathbf{v}_1 = a\left( \frac{\sqrt{3}}{2},\frac{1}{2}\right)\) and \(\mathbf{v}_2 = a \left( \frac{\sqrt{3}}{2},-\frac{1}{2}\right)\) with \(a>0\). Then \(\Lambda=\mathbb{Z}\mathbf{v}_1 \oplus \mathbb{Z}\mathbf{v}_2\) is a triangular lattice with the fundamental domain \(\Gamma=\mathbb{R}^2/\Lambda\). Note that the fundamental domain \(\Gamma\) contains only one lattice point at \(\mathbf{x}=\mathbf{0}\).

The union of two triangular lattices generates a honeycomb lattice \begin{equation}\label{H} H= \Lambda \cup (\Lambda+\mathbf{x}_0)\end{equation} where \(\mathbf{x}_0 = \frac{2}{3} (\mathbf{v}_1+\mathbf{v}_2)\). Note that \(\Gamma\) contains two points of \(H\) at \(\mathbf{x}=\mathbf{0}\) and \(\mathbf{x}=\mathbf{x}_0\).

\begin{figure}
\centering
\begin{subfigure}{0.47\textwidth}
\centering
\includegraphics[width=\textwidth]{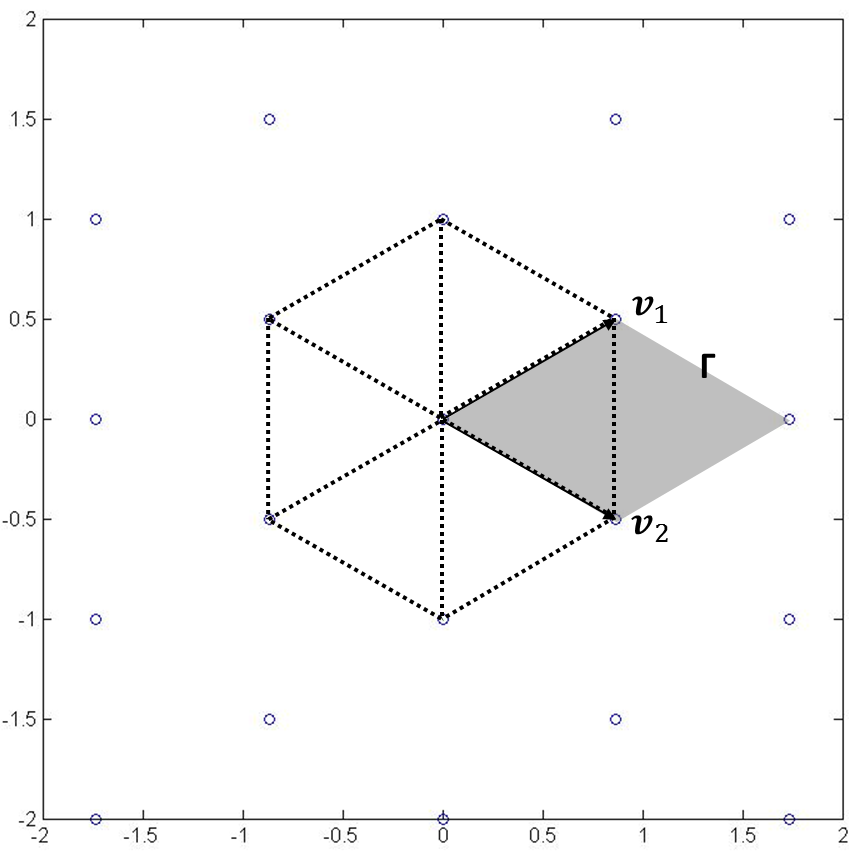} 
\caption{Triangular lattice }
\label{tri}
\end{subfigure}
\begin{subfigure}{0.47\textwidth}
\centering
\includegraphics[width=\textwidth]{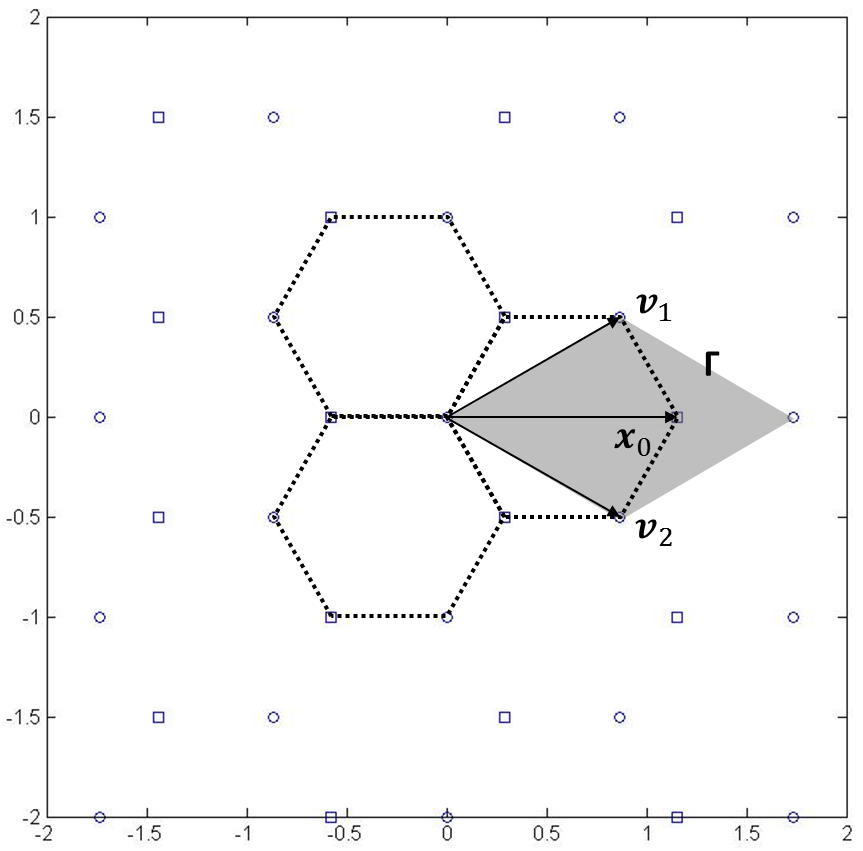} 
\caption{Honeycomb lattice}
\label{hex}
\end{subfigure}
\caption{Two kinds of lattice structures. \(\Gamma\) is the fundamental domain of both lattices. The translated triangular lattice points on the honeycomb lattice are marked as \(\times\).
}\label{trihex}
\end{figure}

The dual lattice \(\Lambda^*=\mathbb{Z}\mathbf{k}_1 \oplus \mathbb{Z}\mathbf{k}_2\) corresponding to the fundamental domain \(\Gamma\) is spanned by two vectors \(\mathbf{k}_1, \mathbf{k}_2\) satisfying
\[\mathbf{k}_i\cdot \mathbf{v}_j =2\pi \delta_{ij}, \quad i,j=1,2\]
or equivalently,
\[\mathbf{k}_1=\frac{4\pi}{a\sqrt{3}}\left( \frac{1}{2}, \frac{\sqrt{3}}{2} \right), \quad \mathbf{k}_2=\frac{4\pi}{a\sqrt{3}}\left( \frac{1}{2}, -\frac{\sqrt{3}}{2} \right),\]
A Brillouin Zone \(\mathcal{B}\) is defined as a hexagon centered at the origin. (See Figure \ref{dual}.)
Note that the triangular lattice and the honeycomb lattice share the hexagonal Brillouin zone since they have the same fundamental domain \(\Gamma\). The six vertices of \(\mathcal{B}\) are called Dirac points. 

\begin{figure}
\centering
\includegraphics[width=0.5\columnwidth]{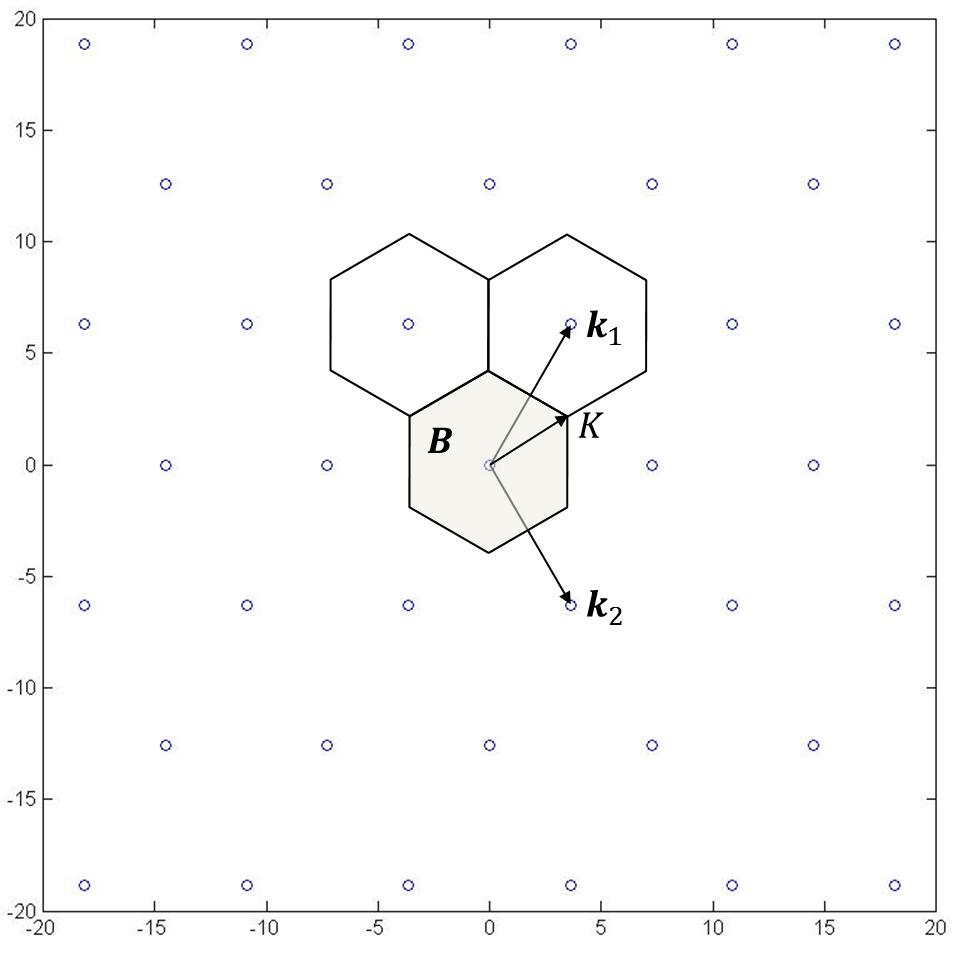}
\caption{Brillouin Zone \(\mathcal{B}\) corresponding to two lattices in Figure~\ref{trihex}. Dual lattice points and Dirac points are marked as \(\times\) and \(\circ\), respectively. In this paper, we only consider \(\mathbf{K}=\frac{2}{3}\mathbf{k}_1+\frac{1}{3}\mathbf{k}_2\) out of six Dirac points without loss of generality.}
\label{dual}
\end{figure}
\FloatBarrier
\subsection{Notation}
\begin{enumerate}
\item \((f,g)=\int_{\Gamma}\overline{f(\mathbf{x})}g(\mathbf{x})d\mathbf{x},\quad f,g \in L^2(\Gamma)\)
\item \(\mult(\lambda,P)\) is the multiplicity of the eigenvalue \(\lambda\) for the operator \(P\).
\item \(\mathbf{k}\) is a vector in Brillouin Zone \(\mathcal{B}\).
\item \(\mathbf{v}_\mathbf{m}= \mathbf{v}_{(m_1,m_2)}=m_1\mathbf{v}_1+m_2\mathbf{v}_2\) is a vector in the triangular lattice \(\Lambda\).
\item \(\xi_\mathbf{m}=\xi_{(m_1,m_2)}= m_1\mathbf{k}_1+m_2\mathbf{k}_2\) is a vector in the dual lattice \(\Lambda^*\).
\item \(L_\mathbf{k}^2(\Gamma)=\{f\in L^2(\Gamma) ~|~ f(\mathbf{x}+\mathbf{v}_j)=e^{i\mathbf{k} \cdot \mathbf{v}_j}f(\mathbf{x}), \quad j=1,2\}\)
\item \(H_\mathbf{k}^m(\Gamma)=\{f\in H^m(\Gamma) ~|~ f(\mathbf{x}+\mathbf{v}_j)=e^{i\mathbf{k} \cdot \mathbf{v}_j}f(\mathbf{x}), \quad j=1,2\}\)
\item \(\Delta(\mathbf{k}) \) is the Laplacian with Floquet boundary condition: \(D(\Delta(\mathbf{k}))=H_\mathbf{k}^2(\Gamma)\)

\item \(\Delta_{\boldsymbol\beta,X}\) is the self-adjoint operator in \(L^2(\mathbb{R}^2)\) for point scatterers placed on \(X=\{\mathbf{x}_1,\mathbf{x}_2,\cdots\}\subset\mathbb{R}^2\) with parameters \(\boldsymbol\beta=(\beta_1,\beta_2,\cdots ),\) where \( \beta_j\in(-\infty,\infty],~ j=1,2,\cdots\). 
\item \(\Delta_{\boldsymbol\alpha,Y}(\mathbf{k})\) is the decomposed self-adjoint operator in \(L^2_\mathbf{k}(\Gamma )\) for periodic point scatterers placed on \(Y=\{\mathbf{y}_1,\cdots,\mathbf{y}_N\}\subset\Gamma\) with parameters \(\boldsymbol\alpha=(\alpha_1,\cdots\alpha_N) \in (-\infty,\infty]^N\). 
\end{enumerate}

In (9), \(\Delta_{\boldsymbol\beta,X}\) is formally defined as a Schrödinger operator \[-\Delta+\sum_{j=1}^\infty c_j \delta({\mathbf{x}-\mathbf{x}_j}),\]
on \(L^2(\mathbb{R}^2)\) where \(c_j\) is constant and \(\delta(\mathbf{x})\) is the Dirac delta function supported at \(\mathbf{0}\in\mathbb{R}^2\). However, a renormalization process is required to make this a self-adjoint operator on a Hilbert space dense in \(L^2(\mathbb{R}^2)\).  Therefore, we construct the periodic point scatterers with the following renormalization process: 
First, consider a finite subset \(X_M \subset X\) consisting of \(M\) points. Then we 
restrict the domain of the Laplacian to the functions vanishing on \(X_M\subset\mathbb{R}^2\). According to the theory of self-adjoint extension by von Neumann, such a symmetric operator has a family of self-adjoint extensions parametrized by  \(\beta_j\in (-\infty,\infty],~j=1,\cdots, M\). Then we obtain \(\Delta_{\boldsymbol\beta,X}\) as the limit case as \(M \rightarrow \infty\) in the norm resolvent sense. We are following Albeverio's notation \cite{solvable} so that the point scatterer at \(\mathbf{x}_j\) gets stronger when \(|\beta_j|\ll \infty\) and it disappears as \(\beta_j\rightarrow \infty \). In particular, when \(\beta_j=\infty\), \(\Delta_{\boldsymbol\beta,X}\) acts as the Laplacian in the neighborhood of \(\mathbf{x}_j\).

In (10), we may follow the same steps as in (9) to define \(\Delta_{\boldsymbol\alpha,Y}(\mathbf{k})\) as a renormalization of 
\[-\Delta+\sum_{j=1}^N c_j \delta({\mathbf{x}-\mathbf{y}_j})\] on \(L_\mathbf{k}^2(\Gamma)\).
The Floquet Laplacian restricted to the functions vanishing at \(Y\subset\Gamma\) has a family of self-adjoint extension \(\Delta_{\boldsymbol\alpha,Y}(\mathbf{k})\) with parameters \(\alpha_j\in (-\infty,\infty]\) determining the strength of the point scatterer at \(\mathbf{y}_j\). 
See Appendix~\ref{sec:introtoptsctr} for the rigorous descriptions of \(\Delta_{\boldsymbol\beta,X}\) and \(\Delta_{\boldsymbol\alpha,Y}(\mathbf{k})\).

\section{Honeycomb lattice point scatterers}
Consider periodic point scatterers placed on a triangular lattice \(\Lambda\) and on a honeycomb lattice \(H\). We will follow the notion of point scatterers defined in Appendix~\ref{sec:introtoptsctr} in which we introduce Floquet theory and construct the point scatterers on \(\mathbb{R}^2\) using renormalization and the theory of self-adjoint extensions. For periodic point scatterers, there exists the decomposition of an operator in \(L^2(\mathbb{R}^2)\) into ``fibers'' \cite{reedsimon4} in \(L^2_\mathbf{k}(\Gamma)\) that corresponds to the Floquet theory for a Schrödinger operator \(-\Delta+V\) with a periodic potential \(V\). See also Part III of \cite{solvable} for more details.
Note that the resolvent formula of \(\Delta_{\boldsymbol\alpha,Y}(\mathbf{k})\) and the decomposition of \(\Delta_{\boldsymbol\beta,\Lambda+Y}\) introduced in Proposition~\ref{ndelta} and \ref{prop:fiber}, respectively, play crucial role to study spectral properties of periodic point scatterers.

In this section, we first compare the decomposed operator for the periodic point scatterer \(-\Delta_{\boldsymbol\alpha,Y}(\mathbf{k})\) to a decomposed linear Schrödinger operator \(-\Delta(\mathbf{k})+V\) when \(\Lambda+Y\), the set of scattering points, and the potential \(V\) share some symmetries introduced in \cite{honeycomb}. 
\begin{defi}[Honeycomb lattice potentials]
\(V\in C^\infty (\mathbb{R}^2;\mathbb{R})\) is a honeycomb lattice potential if there exists \(\mathbf{y}\in\Gamma \) such that 
\begin{enumerate}
\item \(V\) is \(\mathcal{T}_{\mathbf{v}}\)-invariant for all \(\mathbf{v}\in \Lambda\) where \(\mathcal{T}_{\mathbf{v}}V(\mathbf{x})=V(\mathbf{x}+\mathbf{v})\) is the translation along \(\mathbf{v}\), i.e. \(V(\mathbf{x}+\mathbf{v}) = V(\mathbf{x})\)
\item \(V\) is \(\mathcal{I}_{\mathbf{y}}\)-invariant where \(\mathcal{I}_{\mathbf{y}} V(\mathbf{x})= V(2 \mathbf{y}-\mathbf{x})\) is the inversion with respect to \(\mathbf{y}\), i.e. \(V(2\mathbf{y}-\mathbf{x})=V(\mathbf{x})\)
\item \(V\) is \(\mathcal{R}_{\mathbf{y}}\)-invariant where \(\mathcal{R}_{\mathbf{y}}V(\mathbf{x})=V(R(\mathbf{x}-\mathbf{y})+\mathbf{y})\) is the \(\frac{2\pi}{3}\)-rotation with respect to \(\mathbf{y}\) with \begin{equation}\label{R}
R=\frac{1}{2}\begin{bmatrix}-1 & -\sqrt{3} \\ \sqrt{3}& -1\end{bmatrix}\end{equation}
i.e. \(V(R(\mathbf{x}-\mathbf{y})+\mathbf{y})=V(\mathbf{x})\)
\end{enumerate}
\end{defi}

\begin{rmk} Since the Laplacian is invariant under \(\mathcal{T}_{\mathbf{v}}, \mathcal{I}_{\mathbf{y}}\) and \(\mathcal{R}_{\mathbf{y}}\), whenever \(V\) is a honeycomb lattice potential, these operations affect only the Floquet boundary condition. In other words, they modify the domain of the operator but not the action of the operator itself. 
\begin{align*}
\mathcal{T}_{\mathbf{v}}(-\Delta(\mathbf{k})+V) f &= (-\Delta(\mathbf{k})+V) \mathcal{T}_{\mathbf{v}} f, \quad \mathbf{v}\in\Lambda\\
\mathcal{I}_{\mathbf{y}}(-\Delta(\mathbf{k})+V) f &= (-\Delta(-\mathbf{k})+V) \mathcal{I}_{\mathbf{y}} f \\
\mathcal{R}_{\mathbf{y}}(-\Delta(\mathbf{k})+V) f &= (-\Delta(R^*\mathbf{k})+V) \mathcal{R}_{\mathbf{y}} f 
\end{align*}
\end{rmk}

Similarly, we can construct a decomposed operator of periodic point scatterers \(-\Delta_{\boldsymbol\alpha,X}(\mathbf{k})\) with the same properties. Suppose that \(X=\{\mathbf{x}_1,\cdots,\mathbf{x}_N \} \subset \Gamma\) and there exists \(\mathbf{y}\in \Gamma\) such that \(\Lambda+X\) is inversion-symmetric and \(\frac{2\pi}{3}\)-rotation invariant with respect to \(\mathbf{y}\), i.e. \begin{equation}\label{inversion1}
\mathbf{y}-(\Lambda+X)=(\Lambda+X)-\mathbf{y}
\end{equation} and
\begin{equation}\label{rotation1}
R\left(\left(\Lambda+X\right)-\mathbf{y}\right)=(\Lambda+X)-\mathbf{y}
\end{equation} where \(R\) is given by \eqref{R}. For example, a triangular lattice \(\Lambda=\Lambda+\{\mathbf{0}\}\) and a honeycomb lattice \(H= \Lambda + \{\mathbf{0},\mathbf{x}_0\}\) satisfy those conditions where \(\mathbf{y}=\frac{1}{3}(\mathbf{v}_1+\mathbf{v}_2)\).

In addition, suppose that for each \(j=1,\cdots,N\), \(\alpha_j\in(-\infty,\infty]\) is invariant under the inversion and the \(\frac{2\pi}{3}\)-rotation with respect to \(\mathbf{y}\), i.e., 
\begin{equation}\label{inversion2}
2\mathbf{y}-\mathbf{x}_j-\mathbf{x}_{j'}\in\Lambda ~\text{ implies }~ \alpha_j=\alpha_{j'}.
\end{equation} and 
\begin{equation}\label{rotation2}
R(\mathbf{x}_j-\mathbf{y})+\mathbf{y}-\mathbf{x}_{j'}\in \Lambda ~\text{ implies }~ \alpha_j=\alpha_{j'}.\end{equation}

\begin{prop}\label{honeycombsymm}
Let \(X=\{\mathbf{x}_1,\cdots,\mathbf{x}_N \} \subset \Gamma\). Suppose that \eqref{inversion1}, \eqref{rotation1}, \eqref{inversion2}, and \eqref{rotation2} hold for some \(\mathbf{y}\in\Gamma\). Then for \(\psi\in D(-\Delta_{\boldsymbol\alpha,X}(\mathbf{k}))\), we have
\begin{align*}
\mathcal{T}_{\mathbf{v}} \Delta_{\boldsymbol\alpha,X}(\mathbf{k}) \psi &= \Delta_{\boldsymbol\alpha,X}(\mathbf{k}) \mathcal{T}_{\mathbf{v}} \psi,\quad \mathbf{v}\in \Lambda
\\ \mathcal{I}_{\mathbf{y}}\Delta_{\boldsymbol\alpha,X}(\mathbf{k}) \psi &= \Delta_{\boldsymbol\alpha,X}(-\mathbf{k}) \mathcal{I}_{\mathbf{y}} \psi
\\ \mathcal{R}_{\mathbf{y}}\Delta_{\boldsymbol\alpha,X}(\mathbf{k}) \psi &= \Delta_{\boldsymbol\alpha,X}(R^*\mathbf{k}) \mathcal{R}_{\mathbf{y}} \psi
\end{align*}
\end{prop}
\begin{proof} Let \(\psi\in D(-\Delta_{\boldsymbol\alpha,X}(\mathbf{k}))\) and choose any \(\lambda \notin \sigma(-\Delta(\mathbf{k}))\cup \sigma(-\Delta_{\boldsymbol\alpha,X}(\mathbf{k}))\). We can define \(\phi_\lambda \in H_\mathbf{k}^2(\Gamma)\) as in \eqref{ndeltadomain}. Since \(g_\lambda(\bullet,\mathbf{k})\in L_\mathbf{k}^2(\Gamma)\) , for \(\mathbf{v} \in \Lambda\), we have
\begin{align*}
&\mathcal{T}_{\mathbf{v}}(-\Delta_{\boldsymbol\alpha,X}(\mathbf{k})-\lambda)\psi =\mathcal{T}_{\mathbf{v}}(-\Delta(\mathbf{k})-\lambda)\phi_\lambda\\
&=(-\Delta(\mathbf{k})-\lambda)\mathcal{T}_{\mathbf{v}}\phi_\lambda
=(-\Delta_{\boldsymbol\alpha,X}(\mathbf{k})-\lambda)\mathcal{T}_{\mathbf{v}}\psi.
\end{align*}

For the inversion property, we observe 
\begin{align*}
&\mathcal{I}_{\mathbf{y}}(-\Delta_{\boldsymbol\alpha,X}(\mathbf{k})-\lambda)\psi =\mathcal{I}_{\mathbf{y}}(-\Delta(\mathbf{k})-\lambda)\phi_\lambda\\
&=(-\Delta(-\mathbf{k})-\lambda)\mathcal{I}_{\mathbf{y}}\phi_\lambda=(-\Delta_{\boldsymbol\alpha,X}(-\mathbf{k})-\lambda) \\&\quad\left[\mathcal{I}_{\mathbf{y}}\phi_\lambda+\frac{1}{\mathrm{area}(\mathcal{B})} \sum_{j,j'=1}^{N}{\left[\Gamma_{\boldsymbol\alpha,X} (\lambda,\mathbf{k})^{-1}\right]_{jj'} (\mathcal{I}_{\mathbf{y}}\phi_\lambda)(\mathbf{x}_j ) g_\lambda \left(\mathbf{x}-\mathbf{x}_{j'}, \mathbf{k}\right)}\right]
\\&= (-\Delta_{\boldsymbol\alpha,X}(-\mathbf{k})-\lambda)
\\&\quad\left[\mathcal{I}_{\mathbf{y}}\phi_\lambda+\frac{1}{\mathrm{area}(\mathcal{B})} \sum_{l,l'=1}^{N}{\left[\Gamma_{\boldsymbol\alpha,X} (\lambda,-\mathbf{k})^{-1}\right]_{ll'} \phi_\lambda(\mathbf{x}'_l ) (\mathcal{I}_{\mathbf{y}'} g_\lambda) \left(\mathbf{x}-\mathbf{x}'_{l'}, -\mathbf{k}\right)}\right]
\\&=(-\Delta_{\boldsymbol\alpha,X}(-\mathbf{k})-\lambda)\mathcal{I}_{\mathbf{y}}\psi
\end{align*}
by rearranging \(X=\{\mathbf{x}_1,\cdots,\mathbf{x}_n\}=\{\mathbf{x}'_1,\cdots,\mathbf{x}'_n\}\) so that \[2\mathbf{y}-\mathbf{x}_j - \mathbf{x}'_l\in \Lambda ~\text{ and }~ 2\mathbf{y}-\mathbf{x}_{j'} - \mathbf{x}'_{l'}\in\Lambda,\quad j,j',l,l'=1,\cdots,N\] 

Similarly, we prove the rotation property as in the inversion case by rearranging \(X=\{\mathbf{x}_1,\cdots,\mathbf{x}_n\}=\{\mathbf{x}'_1,\cdots,\mathbf{x}'_n\}\) so that 
\[
R(\mathbf{x}_j-\mathbf{y})+\mathbf{y}-\mathbf{x}'_{l}\in \Lambda ~\text{ and }~ R(\mathbf{x}_{j'}-\mathbf{y})+\mathbf{y}-\mathbf{x}'_{l'}\in \Lambda,\quad j,j',l,l'=1,\cdots,N
\] 
Hence,
\[\mathcal{R}_{\mathbf{y}}(-\Delta_{\boldsymbol\alpha,X}(\mathbf{k})-\lambda)\psi=
(-\Delta_{\boldsymbol\alpha,X}(R^*\mathbf{k})-\lambda)\mathcal{R}_{\mathbf{y}}\psi\]
This concludes the proof.
\end{proof}

\subsection{Point scatterers on the triangular lattice}\label{sec:triangular}

As a preliminary step toward point scatterers on the honeycomb lattice, we consider point scatterers on the triangular lattice \(\Lambda\). We first summarize several known results for point scatterers on a periodic lattice with only one scatterer in the fundamental domain. See Chapter III.4 of \cite{solvable} for more details. Then we observe the spectral properties of the triangular lattice point scatter as direct applications of those results.

Suppose there is one point scatterer on \(\Gamma\), say at \(\mathbf{0}\). Note that the strength of the point scatterer is determined by \(\alpha\in (-\infty,\infty]\) in terms of the resolvent operator in Proposition~\ref{ndelta}. We will exclude \(\alpha=\infty\) from our consideration since the point scatterers annihilate at \(\alpha=\infty\); therefore, \[-\Delta_{\infty,\{\mathbf{0}\}}(\mathbf{k})=-\Delta(\mathbf{k}).\]
For \(\alpha\in\mathbb{R}\), we can group all the eigenvalues of \(-\Delta_{\alpha,\{\mathbf{0}\}}(\mathbf{k})\) into two categories: the perturbed eigenvalues \[\lambda\in\sigma(-\Delta_{\alpha,\{\mathbf{0}\}}(\mathbf{k}))\setminus\sigma(-\Delta(\mathbf{k}))\] and the unperturbed eigenvalues
\[\lambda\in \sigma(-\Delta_{\alpha,\{\mathbf{0}\}}(\mathbf{k}))\cap\sigma(-\Delta(\mathbf{k})).\]

By the resolvent formula of periodic point scatterers, the perturbed eigenvalues \(\lambda\) of a triangular lattice point scatterer are given as the solutions to \[\Gamma_{\alpha,\{\mathbf{0}\}} (\lambda,\mathbf{k})=0\] where \(\Gamma_{\alpha,\{\mathbf{0}\}} (\lambda,\mathbf{k})\) is defined in \eqref{gamma}. In other words, \(\lambda\) is a perturbed eigenvalue of \(-\Delta_{\alpha,\{\mathbf{0}\}}(\mathbf{k})\) if and only if 
\begin{equation}\label{1deltaglambda}
\alpha=g_\lambda(\mathbf{0},\mathbf{k})
\end{equation} with \(g_\lambda\) in \eqref{glambda}, or equivalently,

\begin{equation} \label{1delta}
\mathrm{area}(\Gamma)\alpha= \lim_{ r \rightarrow \infty}{\left[\sum_{\substack{\mathbf{m}\in\mathbb{Z}^2\\|\xi_\mathbf{m}+\mathbf{k}| \le r}}{\frac{1}{|\xi_\mathbf{m}+\mathbf{k}|^2-\lambda}}-\frac{2\pi}{\mathrm{area}(\mathcal{B})} \ln r\right]}.
\end{equation}
See Figure~\ref{flop1lpltk0} for schematic graph of RHS of \eqref{1delta} as a function of \(\lambda\).

In addition, every perturbed eigenvalue \(\lambda\) satisfies \(\mult(\lambda,-\Delta_{\alpha,\{\mathbf{0}\}}(\mathbf{k}))=1\) with the corresponding eigenfunction \(g_\lambda(\bullet,\mathbf{k})\).

\begin{rmk} We may use a simpler summation formula \eqref{1deltamod} equivalent to the limit and partial sum notation \eqref{1delta}:
\begin{equation}\label{1deltamod}
\mathrm{area}(\Gamma) (\alpha+\alpha_0)= \sum_{m\in\mathbb{Z}^2}\left[{\frac{1}{|\xi_\mathbf{m}+\mathbf{k}|^2-\lambda}}- \frac{|\xi_\mathbf{m}|^2}{|\xi_\mathbf{m}|^4+1}\right]
\end{equation}
with a constant \(\alpha_0\) defined as
\begin{equation}\label{C}
\alpha_0= \lim_{ r \rightarrow \infty}{\left[ \frac{\ln r}{2\pi}-\frac{1}{\mathrm{area}(\Gamma)}\sum_{\substack{m\in\mathbb{Z}^2\\|\xi_\mathbf{m}| \le r}}\frac{|\xi_\mathbf{m}|^2}{|\xi_\mathbf{m}|^4+1}\right]}
.\end{equation}
\end{rmk}


On the other hand, the unperturbed eigenvalues \(\lambda\) are given as the eigenvalues of \(-\Delta(\mathbf{k})\) with \(\mult (\lambda,-\Delta(\mathbf{k})) >1\). More precisely, let \(\mu=\mult (\lambda,-\Delta(\mathbf{k}))\). Then there exist \(\mathbf{m}_1,\cdots,\mathbf{m}_\mu \in \mathbb{Z}^2\) such that
\[\lambda = |\xi_{\mathbf{m}_1}+\mathbf{k}|^2=\cdots=|\xi_{\mathbf{m}_\mu}+\mathbf{k}|^2. \] Then we have
\begin{equation}\label{1deltamult}
\mult (\lambda,-\Delta_{\alpha,\{\mathbf{0}\}}(\mathbf{k})) = \mult (\lambda,-\Delta(\mathbf{k})) -1
\end{equation}
with the corresponding eigenspace
\[\left\{ f: \mathbf{x}\mapsto \sum_{j=1}^\mu c_j e^{i(\xi_{\mathbf{m}_j}+\mathbf{k})\cdot \mathbf{x}}~\middle|~ f(\mathbf{0})=0,~ c_1,\cdots, c_\mu\in\mathbb{C} \right\}\]
In particular, if \(\mult (\lambda,-\Delta(\mathbf{k}))=1\), then \(\lambda \notin \sigma(-\Delta_{\alpha,\{\mathbf{0}\}}(\mathbf{k}))\).

Now we observe the conic singularities of dispersion bands near Dirac point due to the symmetry property of the triangular lattice.
\begin{thm} \label{prop:1detaconic} Let \(\nu_1 (\mathbf{k},\alpha) \le \nu_2 (\mathbf{k},\alpha)\le \cdots\) be the eigenvalues of \(-\Delta_{\alpha,\{\mathbf{0}\}}(\mathbf{k})\). As \(|\mathbf{k}-\mathbf{K}| \rightarrow 0 \), 
\begin{align} \nu_2 (\mathbf{k},\alpha) &= |\mathbf{K}|^2 -\frac{4\pi}{3a} |\mathbf{k}-\mathbf{K}|+o(|\mathbf{k}-\mathbf{K}|) \\ \nu_3 (\mathbf{k},\alpha) &= |\mathbf{K}|^2 +\frac{4\pi}{3a} |\mathbf{k}-\mathbf{K}|+o(|\mathbf{k}-\mathbf{K}|)
\end{align}
\end{thm}
\begin{proof}
We have \(\nu_2(\mathbf{K},\alpha)=\nu_3(\mathbf{K},\alpha)=|\mathbf{K}|^2\) by \eqref{1deltamult}, whereas \(\nu_1(\mathbf{K},\alpha)<|\mathbf{K}|^2\) is obtained by \eqref{1deltaglambda}. 
Consider \(\delta\mathbf{k}\) with \(|\delta\mathbf{k}|\ll 1\) such that \(|\mathbf{K}+\delta\mathbf{k}|, |\mathbf{K}+R\delta\mathbf{k}|\), and \(|\mathbf{K}+R^2\delta\mathbf{k}|\) have distinct values where \(R\) is defined by \eqref{R}. 
Then we decompose the direction vector \(\mathbf{u}=\frac{\delta\mathbf{k}}{|\delta\mathbf{k}|} \in \mathbb{S}^1 \) into
\[\mathbf{u}= u_1\frac{\mathbf{k}_1}{|\mathbf{k}_1|}+u_2\frac{\mathbf{k}_2}{|\mathbf{k}_2|}\]
where 
\begin{equation}\label{u1u2}u_1^2 +u_2^2 -u_1 u_2 =1.\end{equation}
and 
\begin{equation}\label{k1k2}|\mathbf{k}_1|=|\mathbf{k}_2|=\frac{4\pi}{a\sqrt{3}}.\end{equation}
Therefore,
\begin{equation}\label{Rdeltak}\begin{aligned}
|\mathbf{K}+\delta\mathbf{k}|^2&=|\mathbf{K}|^2 + u_1 |\mathbf{k}_1||\delta\mathbf{k}| +|\delta\mathbf{k}|^2 \\
|\mathbf{K}+R\delta\mathbf{k}|^2&=|\mathbf{K}|^2 + (u_2-u_1) |\mathbf{k}_1||\delta\mathbf{k}| +|\delta\mathbf{k}|^2 \\
|\mathbf{K}+R^2\delta\mathbf{k}|^2&=|\mathbf{K}|^2 - u_2 |\mathbf{k}_1||\delta\mathbf{k}| +|\delta\mathbf{k}|^2 \\
\end{aligned}\end{equation}

Suppose \(\lambda'=|\mathbf{K}|^2+\delta\lambda\) solves \eqref{1deltaglambda} at \(\mathbf{k}=\mathbf{K}+\delta\mathbf{k}\), namely,
\begin{equation}\label{1deltadeltak}
\mathrm{area}(\Gamma)\alpha= \lim_{ r \rightarrow \infty}{\left[\sum_{\substack{\mathbf{m}\in\mathbb{Z}^2\\|\xi_\mathbf{m}+\mathbf{K}+\delta\mathbf{k}| \le r}}{\frac{1}{|\xi_\mathbf{m}+\mathbf{K}+\delta\mathbf{k}|^2-\lambda'}}-\frac{2\pi}{\mathrm{area}(\mathcal{B})} \ln r\right]}.
\end{equation}

To simplify the notation, let \[\mathcal{M}_0=\{\mathbf{m}\in\mathbb{Z}^2 ~|~ |\xi_\mathbf{m}+\mathbf{K}|=|\mathbf{K}|\}=\{(0,0),(-1,0),(-1,-1)\}\] and let
\[
C_0= \lim_{ r \rightarrow \infty}{\left[\sum_{\substack{\mathbf{m}\in\mathbb{Z}^2\setminus \mathcal{M}_0 \\|\xi_\mathbf{m}+\mathbf{K}| \le r}}{\frac{1}{|\xi_\mathbf{m}+\mathbf{K}|^2-\lambda'}}-\frac{2\pi}{\mathrm{area}(\mathcal{B})} \ln r\right]}
\]

Then due to the assumption on \(\delta\mathbf{k}\), \eqref{1deltadeltak} reads 
\[\begin{aligned}\mathrm{area}(\Gamma)\alpha-C_0 =&\frac{1}{|\mathbf{K}+\delta\mathbf{k}|^2-\lambda'}+\frac{1}{|\mathbf{K}+R\delta\mathbf{k}|^2-\lambda'}+\frac{1}{|\mathbf{K}+R^2\delta\mathbf{k}|^2-\lambda'} \\
=&\frac{1}{|\mathbf{k}_1||\delta\mathbf{k}| u_1 -|\delta\mathbf{k}|^2 -\delta\lambda} +\frac{1}{|\mathbf{k}_1||\delta\mathbf{k}| (-u_1+u_2) -|\delta\mathbf{k}|^2 -\delta\lambda} \\&+ \frac{1}{|\mathbf{k}_1||\delta\mathbf{k}| (-u_2) -|\delta\mathbf{k}|^2 -\delta\lambda}\\
\end{aligned}\]

Multiplying \(|\delta\mathbf{k}|\) on both sides, we obtain as \(|\delta\mathbf{k}|\rightarrow 0\),
\[\dfrac{1}{|\mathbf{k}_1|u_1- \dfrac{\delta\lambda}{|\delta\mathbf{k}|}}+\dfrac{1}{|\mathbf{k}_1|(u_2-u_1)- \dfrac{\delta\lambda}{|\delta\mathbf{k}|}}+\dfrac{1}{|\mathbf{k}_1|(-u_2)- \dfrac{\delta\lambda}{|\delta\mathbf{k}|}}\rightarrow 0 \]
By \eqref{u1u2} and \eqref{k1k2}, 
\[\lim_{|\delta\mathbf{k}|\rightarrow 0 } \frac{\delta\lambda}{\delta\mathbf{k}} = \pm \frac{|\mathbf{k}_1|}{\sqrt{3}} = \pm \frac{4\pi}{3a}\] 

By the continuity of each dispersion band, we can extend this result to any direction \(\mathbf{u}\in \mathbb{S}^1\). So we obtain the directional derivative of \(\mathbf{k} \mapsto \nu_j(\mathbf{k},\alpha)\) at \(\mathbf{k}=\mathbf{K}\), \(j=2,3\) as follows:
\[
\nabla_{\mathbf{u}} \nu_2(\mathbf{K},\alpha)=-\frac{4\pi}{3a},\qquad
\nabla_{\mathbf{u}} \nu_3(\mathbf{K},\alpha)=+\frac{4\pi}{3a}
\]
\end{proof}


Note that the second and third dispersion bands described in Theorem~\ref{prop:1detaconic} are not the only pair of conic dispersion bands. Using the same argument, we observe infinitely many additional pairs of dispersion bands with a \((3n-1)\)-fold conic singularity at each \((\mathbf{K},\lambda')\) where 
\[\lambda'\in \{|\xi_\mathbf{m} +\mathbf{K}|^2 ~|~\mathbf{m}\in\mathbb{Z}^2\} = \sigma(-\Delta(\mathbf{K}))\]
and 
\begin{equation}\label{n}
n= \frac{1}{3}\left(\#\{\mathbf{m}\in\mathbb{Z}^2 ~|~ |\xi_\mathbf{m}+\mathbf{K}|=\lambda'\}\right)
\end{equation}
Furthermore, these are the only conic singularities at Dirac point since all the other perturbed eigenvalues \(\lambda'\in\sigma(-\Delta_{\alpha,\{\mathbf{0}\}}(\mathbf{K}))\setminus\sigma(-\Delta(\mathbf{K}))\) have multiplicity 1 which does not allow dispersion bands to meet at the eigenvalue.
Note that \(n\ge 1\) in \eqref{n}  is an integer since for any \(\mathbf{m}\in\mathbb{Z}^2\), we have \[|\xi_\mathbf{m}+\mathbf{K}|=|\tilde{R}\xi_\mathbf{m}+\mathbf{K}|=|\tilde{R}^2 \xi_\mathbf{m}+\mathbf{K}|\]
where \(\tilde{R}:\mathbb{R}^2\rightarrow \mathbb{R}^2\) is a \(\frac{2\pi}{3}\)-rotation around a Dirac Point \(-\mathbf{K}\). 
For instance, the \(j\)-th dispersion bands with \(j=2,3\) and \(j=8,\cdots,12\) form a 2-fold and a 5-fold conic singularity, respectively. (See Figure~\ref{floqaloc}.)

Now we investigate the dispersion bands globally over Brillouin zone and the corresponding band spectra. A similar model has been already examined in \cite{solvable} for rectangular two-dimensional lattices with one scatterer in each fundamental domain. Here we introduce an application to the triangular lattice \(\Lambda\). The proof is similar to that of Theorem III.4.7 in \cite{solvable}. 

\begin{prop}\label{prop:floq1spec}
Let \(\nu_1 (\mathbf{k},\alpha) \le \nu_2 (\mathbf{k},\alpha) \le \cdots \) be the eigenvalues of \(-\Delta_{\alpha,\{\mathbf{0}\}}(\mathbf{k})\). For \(\alpha \in \mathbb{R}\) and \(\boldsymbol\beta = (\alpha,\alpha, \cdots)\), the spectrum of the operator \(-\Delta_{\boldsymbol\beta,\Lambda}\) is purely absolutely continuous and equals 
\begin{equation}\sigma(-\Delta_{\boldsymbol\beta,\Lambda})=[\nu_1(\alpha,\mathbf{0}),\nu_1(\alpha,\mathbf{K})] \cup [\nu_2(\alpha),\infty)\end{equation}
where
\begin{equation}\nu_2(\alpha)=\min \left\{\nu_2(\mathbf{0},\alpha), \nu_2 \left(\frac{\mathbf{k}_1+\mathbf{k}_2}{2},\alpha \right)\right\} > 0 \quad \alpha\in \mathbb{R}\end{equation}

In addition, \(\alpha\mapsto\nu_j(\mathbf{k},\alpha)\) is strictly increasing on \(\mathbb{R}\), namely,
\begin{equation}\frac{\partial \nu_j(\mathbf{k},\alpha)}{\partial \alpha} >0,\quad\alpha\in\mathbb{R},\quad \mathbf{k}\in \mathcal{B},\quad j=1,2,\cdots. \end{equation}

Hence, there exists \(\alpha_1 \) such that
\begin{equation}\sigma(-\Delta_{\boldsymbol\beta,\Lambda})=[\nu_1(\alpha,\mathbf{0}),\infty),\quad \alpha \ge \alpha_1\end{equation}
\end{prop}
See Figure~\ref{floq1aglobal} and Supplement Material \#1 for the global view of first five dispersion bands with a fixed \(\alpha\) and various \(\alpha\)'s, respectively.
See also Figure~\ref{floq1spec} for the graph of spectral bands \(\sigma(-\Delta_{\boldsymbol\beta,\Lambda})\) versus \(\alpha\) where \(\boldsymbol\beta=(\alpha,\alpha,\cdots)\).



\subsection{Point scatterers on the honeycomb structure}\label{sec:honeycomb}
Consider the honeycomb structure \(H=\Lambda+Y\) defined in Section~\ref{2lat} with
\[Y=\{0,\mathbf{x}_0\}, \quad \Lambda=\mathbb{Z}\mathbf{v}_1 \oplus \mathbb{Z}\mathbf{v}_2.\] 
Suppose the parameters for two scatterers at \(\mathbf{0}\) and \(\mathbf{x}_0\) are the same, say \(\alpha\). \[\boldsymbol\alpha=(\alpha,\alpha).\] This is a legitimate assumption for the model of crystal structure comprised of only one element, such as carbon atoms in graphene.
Then the set of perturbed eigenvalues, \(\sigma(-\Delta_{\alpha,Y} (\mathbf{k}))\setminus\sigma(-\Delta(\mathbf{k}))\), is determined as follows:


\begin{prop}\label{thm:2deltapert}
Suppose \(\lambda' \notin \sigma(-\Delta(\mathbf{k}))\). Then \(\lambda' \in \sigma(-\Delta_{\alpha,\{\mathbf{0},\mathbf{x}_0\}}(\mathbf{k}))\) if and only if
\begin{equation}
\label{2delta-}
\alpha=g_{\lambda'}(\mathbf{0},\mathbf{k})+|g_{\lambda'}(\mathbf{x}_0, \mathbf{k})|
\end{equation} 
or 
\begin{equation}\label{2delta+}
\alpha=g_{\lambda'}(\mathbf{0},\mathbf{k})-|g_{\lambda'}(\mathbf{x}_0, \mathbf{k})|
\end{equation} 

In addition,
\begin{equation}\label{pertevmult}
\mult (\lambda',-\Delta_{\alpha,Y} (\mathbf{k})) =\dim\ker(\Gamma_{\alpha,\{\mathbf{0},\mathbf{x}_0\}}(\lambda',\mathbf{k}))
\end{equation}
\end{prop}
\begin{proof}
Suppose \(\lambda' \in \sigma(-\Delta_{\alpha,\{\mathbf{0},\mathbf{x}_0\}} (\mathbf{k})) \setminus \sigma(-\Delta(\mathbf{k})) \). By \eqref{gamma}, \(\Gamma_{\boldsymbol\alpha,\{\mathbf{0},\mathbf{x}_0\}}(\lambda,\mathbf{k})\) is Hermitian for all \(\lambda \notin \sigma(-\Delta(\mathbf{k}))\). 
So we can decompose \(\Gamma_{\boldsymbol\alpha,\{\mathbf{0},\mathbf{x}_0\}}(\lambda,\mathbf{k})\) into
\[\Gamma_{\alpha,\{\mathbf{0},\mathbf{x}_0\}}(\lambda,\mathbf{k}) = U_{\alpha,\{\mathbf{0},\mathbf{x}_0\}}(\lambda,\mathbf{k}) \tilde\Gamma_{\alpha,\{\mathbf{0},\mathbf{x}_0\}}(\lambda,\mathbf{k}) U^*_{\alpha,\{\mathbf{0},\mathbf{x}_0\}} (\lambda,\mathbf{k})\] 
where
\begin{equation}\label{gammaunitary}
U_{\alpha,\{\mathbf{0},\mathbf{x}_0\}}(\lambda,\mathbf{k})= 
\begin{dcases}
\frac{1}{\sqrt{2}}\begin{bmatrix}
1 & -\dfrac{g_\lambda (\mathbf{x}_0,\mathbf{k})}{|g_\lambda (\mathbf{x}_0,\mathbf{k})|}\\\dfrac{\overline{g_\lambda} (\mathbf{x}_0,\mathbf{k})}{|g_\lambda (\mathbf{x}_0,\mathbf{k})|}&1
\end{bmatrix} &\mbox{if } g_\lambda(\mathbf{x}_0, \mathbf{k})\ne 0\\\\
\begin{bmatrix}
1&0\\0&1
\end{bmatrix}& \mbox{if } g_\lambda(\mathbf{x}_0, \mathbf{k})=0
\end{dcases}\end{equation}
and 
\begin{equation}\label{gammadiag}\begin{aligned}
\tilde\Gamma_{\alpha,\{\mathbf{0},\mathbf{x}_0\}}(\lambda,\mathbf{k})&=\diag \left( \tilde\gamma_{1,\alpha,\{\mathbf{0},\mathbf{x}_0\}}(\lambda,\mathbf{k}),~\tilde\gamma_{2,\alpha,\{\mathbf{0},\mathbf{x}_0\}}(\lambda,\mathbf{k})\right) \\&=
\diag \left(
\alpha-g_\lambda(\mathbf{0},\mathbf{k})+|g_\lambda(\mathbf{x}_0, \mathbf{k})| ,~\alpha-g_\lambda(\mathbf{0},\mathbf{k})-|g_\lambda(\mathbf{x}_0, \mathbf{k})| \right)
\end{aligned}\end{equation}

Then for \(\lambda \notin \sigma(-\Delta_{\alpha,\{\mathbf{0},\mathbf{x}_0\}} (\mathbf{k}))\), we can rewrite \eqref{infdelta} as
\begin{multline}\label{infdeltadiag}
\left( -\Delta_{\alpha,\{\mathbf{0},\mathbf{x}_0\}}(\mathbf{k})-\lambda \right)^{-1} f(\mathbf{x}) \\ = \left(-\Delta (\mathbf{k})-\lambda \right)^{-1}f(\mathbf{x}) + \frac{1}{\mathrm{area}(\mathcal{B})}\sum_{j=1}^{N}{\tilde\gamma^{-1}_{j,\alpha,\{\mathbf{0},\mathbf{x}_0\}}(\lambda,\mathbf{k}) \left(\overline{\tilde g_{j,\lambda} \left(\bullet,\mathbf{k} \right)},f\right) \tilde g_{j,\lambda} \left(\mathbf{x}, \mathbf{k}\right)}
\end{multline}
where 
\(\tilde g_{j,\lambda}(\mathbf{x},\mathbf{k})\) is the \(j\)-th entry of the vector
\[U^*_{\alpha,\{\mathbf{0},\mathbf{x}_0\}} (\lambda,\mathbf{k}) 
\begin{bmatrix}
g_\lambda \left(\mathbf{x}, \mathbf{k}\right) \\
g_\lambda \left(\mathbf{x}-\mathbf{x}_0, \mathbf{k}\right) 
\end{bmatrix}.\]

In addition, we observe from \eqref{gammadiag} that \[\partial_\lambda \tilde\gamma_{j,\alpha,\{\mathbf{0},\mathbf{x}_0\}}(\lambda,\mathbf{k})< 0, \quad j=1,2\] 
which implies for \(\lambda' \in \sigma(-\Delta_{\alpha,\{0,\mathbf{x}_0\}} (\mathbf{k})) \setminus \sigma(-\Delta(\mathbf{k}))\),
\[\tilde\gamma_{j,\alpha,\{\mathbf{0},\mathbf{x}_0\}}(\lambda',\mathbf{k})=0 \quad\text{ if and only if }\quad \mathrm{Res}_{\lambda=\lambda'}\left(\tilde\gamma^{-1}_{j,\alpha,\{\mathbf{0},\mathbf{x}_0\}}(\lambda,\mathbf{k})\right)\ne 0, \quad j=1,2.\]
Therefore, we obtain the multiplicity of \(\lambda=\lambda'\) as follows: For \(\epsilon >0\) sufficiently small,
\[\begin{aligned}
&\mult (\lambda',-\Delta_{\alpha,\{\mathbf{0},\mathbf{x}_0\}} (\mathbf{k})) \\&= \mathrm{rank}\oint_{|\lambda-\lambda'|=\epsilon}\left( -\Delta_{\alpha,\{\mathbf{0},\mathbf{x}_0\}}(\mathbf{k})-\lambda \right)^{-1}d\lambda\\
&= \mathrm{rank}\oint_{|\lambda-\lambda'|=\epsilon} \frac{1}{\mathrm{area}(\mathcal{B})}\sum_{j=1}^{2}{\tilde\gamma^{-1}_{j,\alpha,\{\mathbf{0},\mathbf{x}_0\}}(\lambda,\mathbf{k}) \left(\overline{\tilde g_{j,\lambda} \left(\mathbf{k} \right)},\bullet\right) \tilde g_{j,\lambda} \left(\mathbf{k}\right)d\lambda} \\
&= \# \left\{j=1, 2 ~\middle|~ \mathrm{Res}_{\lambda=\lambda'}\left(\tilde\gamma^{-1}_{j,\alpha,\{\mathbf{0},\mathbf{x}_0\}}(\lambda,\mathbf{k})\right) \ne 0\right\}\\
&= \# \left\{j=1, 2 ~\middle|~ \tilde\gamma_{j,\alpha,\{\mathbf{0},\mathbf{x}_0\}}(\lambda',\mathbf{k})=0\right\}\\
&=\dim\ker(\Gamma_{\alpha,\{\mathbf{0},\mathbf{x}_0\}}(\lambda',\mathbf{k}))
\end{aligned}\]

In addition, the corresponding eigenspace is the range of the operator:
\begin{equation} \label{pertes}
\mathrm{span}\left\{\tilde g_{j,\lambda'}(\bullet,\mathbf{k}) ~\middle|~\tilde\gamma_{j,\alpha,\{\mathbf{0},\mathbf{x}_0\}}(\lambda',\mathbf{k})=0, ~j=1, 2\right\}
\end{equation}
\end{proof}

\begin{figure}
\centering
\begin{subfigure}{0.47\textwidth}
\centering
\includegraphics[width=\textwidth]{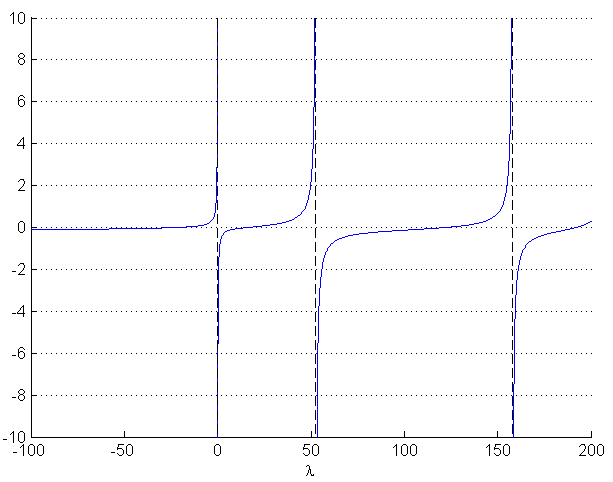} 
\caption{RHS of \eqref{1delta}}
\label{flop1lpltk0}
\end{subfigure}
\begin{subfigure}{0.47\textwidth}
\centering
\includegraphics[width=\textwidth]{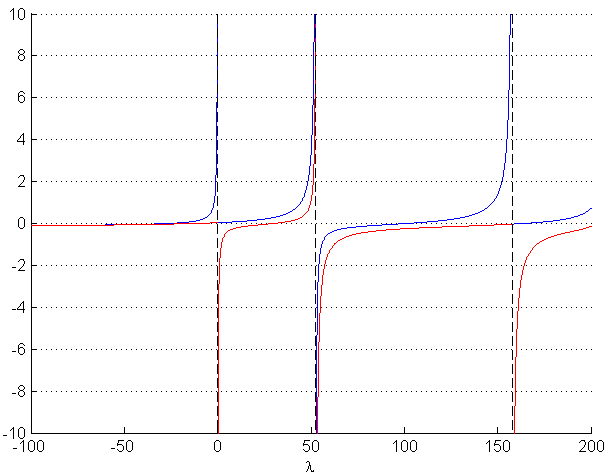}
\caption{RHS of \eqref{2delta-} (dashed) and \eqref{2delta+} (solid)}
\label{flop2lpltk0}
\end{subfigure}
\caption{RHS of \eqref{1delta}, \eqref{2delta-} and \eqref{2delta+} as a function of \(\lambda\). The dashed vertical lines represent \(\lambda \in \sigma(-\Delta(\mathbf{k}))\).}
\label{floplpltk0}
\end{figure}

On the other hand, we observe that some eigenvalues of \(-\Delta(\mathbf{k}) \) remain in the spectrum of \(-\Delta_{\alpha,\{0,\mathbf{x}_0\}}(\mathbf{k}) \) as unperturbed eigenvalues with multiplicity decreased by 0, 1 or 2.
\begin{prop}\label{thm:2deltaunpert}
Suppose \(\lambda'=\lambda_{\mathbf{m}_1} = \cdots =\lambda_{\mathbf{m}_\mu} \) is an eigenvalue of the unperturbed operator \(-\Delta(\mathbf{k}) \) of multiplicity \(\mu\) where \(\lambda_\mathbf{m}=|\xi_\mathbf{m}+\mathbf{k}|^2 \). Then 
\begin{multline}\label{2deltamult}
\mult (\lambda',-\Delta_{\alpha,\{\mathbf{0},\mathbf{x}_0\}}(\mathbf{k})) =\\
\begin{cases}
\mu -2 &\mbox{if } \mu \ne \left| \sum_{j=1}^\mu e^{i\xi_{\mathbf{m}_j} \cdot \mathbf{x}_0}\right| \text{ (Case 1)}\\\\
\mu-1 &\mbox{if } \mu = \left| \sum_{j=1}^\mu e^{i\xi_{\mathbf{m}_j} \cdot \mathbf{x}_0}\right| \\&\text{and } \alpha \ne \lim_{\lambda\nearrow \lambda'} \left(g_\lambda(\mathbf{0},\mathbf{k})-\left|g_\lambda(\mathbf{x}_0,\mathbf{k})\right|\right) \text{ (Case 2)}\\\\
\mu &\mbox{if } \mu = \left| \sum_{j=1}^\mu e^{i\xi_{\mathbf{m}_j} \cdot \mathbf{x}_0}\right| \\&\text{and } \alpha = \lim_{\lambda\nearrow \lambda'} \left(g_\lambda(\mathbf{0},\mathbf{k})-\left|g_\lambda(\mathbf{x}_0,\mathbf{k})\right|\right) \text{ (Case 3)}
\end{cases}
\end{multline}
with the corresponding eigenspaces in the same order
\begin{equation}\label{es1}
\left\{ f: \mathbf{x}\mapsto \sum_{j=1}^\mu c_j e^{i(\xi_{\mathbf{m}_j}+\mathbf{k})\cdot \mathbf{x}}~\middle|~ f(\mathbf{0})=f(\mathbf{x}_0)=0,~ c_1,\cdots, c_\mu\in\mathbb{C} \right\} \end{equation}\begin{equation} \label{es2}
\left\{ f:\mathbf{x}\mapsto \sum_{j=1}^\mu c_j e^{i(\xi_{\mathbf{m}_j}+\mathbf{k})\cdot \mathbf{x}}~\middle|~ f(\mathbf{0}) =0,~ c_1,\cdots, c_\mu\in\mathbb{C}\right\} \end{equation}\begin{equation}\label{es3}
\left\{ f:\mathbf{x}\mapsto \sum_{j=1}^\mu c_j e^{i(\xi_{\mathbf{m}_j}+\mathbf{k})\cdot \mathbf{x}}~\middle|~ f(\mathbf{0}) =0,~ c_1,\cdots, c_\mu\in\mathbb{C}\right\} \oplus \mathrm{span}\left\{\tilde g_{2,\lambda'_-}(\bullet,\mathbf{k})\right\}
\end{equation} with \(\tilde g_{2,\lambda'_-}(\bullet,\mathbf{k})=\lim_{\lambda\nearrow\lambda'}\tilde g_{2,\lambda}(\bullet,\mathbf{k})\) as in \eqref{pertes}.
Note that \[\mult (\lambda',-\Delta_{\alpha,\{0,\mathbf{x}_0\}}(\mathbf{k}))=0\] means \(\lambda' \notin \sigma(-\Delta_{\alpha,\{0,\mathbf{x}_0\}}(\mathbf{k}))\).
\end{prop}
\begin{proof} First, consider the Laurent expansion of \(g_\lambda \) as \(\lambda \rightarrow \lambda'\). 
\[\mathrm{area}(\Gamma) g_\lambda (\mathbf{x},\mathbf{k})=-\sum_{j=1}^\mu e^{i\xi_{\mathbf{m}_j} \cdot \mathbf{x}_0}(\lambda-\lambda')^{-1}+R_{\lambda'}(\mathbf{x},\mathbf{k})+O(|\lambda-\lambda'|),\quad \lambda\rightarrow\lambda'\] 
with the remainder term \(R_{\lambda}(\mathbf{x},\mathbf{k})=O(1)\). Also, suppose \(f\in L_\mathbf{k}^2(\Gamma)\) has the expansion
\[f(\mathbf{x})=\sum_{\mathbf{m}\in\mathbb{Z}^2}f_\mathbf{m} e^{i(\xi_\mathbf{m}+\mathbf{k})\cdot \mathbf{x}},\quad f_\mathbf{m}=\frac{1}{\mathrm{area}(\mathcal{B})}\int_\mathcal{B} f(\mathbf{x})e^{-i(\xi_\mathbf{m}+\mathbf{k})\cdot\mathbf{x}}d\mathbf{x}\]
Then we can rewrite the integral kernel of \eqref{infdelta} for \(-\Delta_{\alpha,\{\mathbf{0},\mathbf{x}_0\}}(\mathbf{k})\) using the matrix and vector notation.
\begin{multline}\label{kernel}
\left(-\Delta_{\alpha,\{\mathbf{0},\mathbf{x}_0\}}(\mathbf{k})-\lambda \right)^{-1}(\mathbf{x},\mathbf{x}') =\\ g_\lambda(\mathbf{x}-\mathbf{x}') + \frac{1}{\mathrm{area}(\mathcal{B})}\left[\overline{\vec{g}_\lambda} (\mathbf{x}',\mathbf{k})\right]^T \left[\Gamma_{\alpha,\{\mathbf{0},\mathbf{x}_0\}} (\lambda,\mathbf{k})\right]^{-1} \left[ \vec{g}_\lambda(\mathbf{x},\mathbf{k})\right]
\end{multline}
where each term has the Laurent expansion
\[\vec{g}_\lambda (\mathbf{x},\mathbf{k})=\begin{bmatrix}-\sum_{j=1}^\mu e^{i(\xi_{\mathbf{m}_j}+\mathbf{k})\cdot\mathbf{x}} \\ -\sum_{j=1}^\mu e^{i(\xi_{\mathbf{m}_j}+\mathbf{k})\cdot(\mathbf{x}-\mathbf{x}_0)}\end{bmatrix}(\lambda-\lambda')^{-1}+
\begin{bmatrix}
R_{\lambda'}(\mathbf{x},\mathbf{k}) \\
R_{\lambda'}(\mathbf{x}-\mathbf{x}_0,\mathbf{k})
\end{bmatrix} + O(|\lambda-\lambda'|)\]

\begin{multline*}
\mathrm{adj}~\Gamma_{\alpha,\{\mathbf{0},\mathbf{x}_0\}} (\lambda,\mathbf{k})= 
\begin{bmatrix}
\mu & -\sum_{j=1}^\mu e^{i(\xi_{\mathbf{m}_j}+\mathbf{k})\cdot\mathbf{x}_0}\\
-\sum_{j=1}^\mu e^{-i(\xi_{\mathbf{m}_j}+\mathbf{k})\cdot\mathbf{x}_0}& \mu
\end{bmatrix} (\lambda-\lambda')^{-1}\\+
\begin{bmatrix}
\alpha-R_{\lambda'}(\mathbf{0},\mathbf{k}) & R_{\lambda'}(\mathbf{x}_0,\mathbf{k})\\
R_{\lambda'}(-\mathbf{x}_0,\mathbf{k}) & \alpha-R_{\lambda'}(\mathbf{0},\mathbf{k})
\end{bmatrix} +O(|\lambda-\lambda'|)
\end{multline*}

\[\begin{aligned}
\det\Gamma_{\alpha,\{\mathbf{0},\mathbf{x}_0\}} (\lambda,\mathbf{k})=&
\left(\mu^2-\left| \sum_{j=1}^\mu e^{i\xi_{\mathbf{m}_j} \cdot \mathbf{x}_0}\right|^2\right)(\lambda-\lambda')^{-2}
\\&+ 2\left(\mu(\alpha-R_{\lambda'}(\mathbf{0}))+\Re\left( \sum_{j=1}^\mu e^{-i\xi_{\mathbf{m}_j} \cdot \mathbf{x}_0} R_{\lambda'}(\mathbf{x}_0)\right)\right)(\lambda-\lambda')^{-1}\\&+O(1)
\\=& C_{-2}(\lambda-\lambda')^{-2}+C_{-1}(\lambda-\lambda')^{-1}+C_0+O(|\lambda-\lambda'|)
\end{aligned}\]

\emph{Case 1:} Suppose \(\mu \ne \left| \sum_{j=1}^\mu e^{i\xi_{\mathbf{m}_j} \cdot \mathbf{x}_0}\right|\). Consider an operator \(P:L_\mathbf{k}^2 (\Gamma)\rightarrow L_\mathbf{k}^2(\Gamma)\) as the norm limit \[P=\lim_{\lambda\rightarrow\lambda'}(\lambda-\lambda')\left(-\Delta_{\alpha,\{\mathbf{0},\mathbf{x}_0\}}(\mathbf{k})-\lambda \right)^{-1}.\]
Then \(P\) is a projection onto the eigenspace corresponding to the eigenvalue \(\lambda=\lambda'\).
Since \(C_{-2}\ne 0\), for all \(f\in L_\mathbf{k}^2(\Gamma)\),
\begin{multline*}
Pf(\mathbf{x}) =
\sum_{j=1}^\mu e^{i(\xi_{\mathbf{m}_j}+\mathbf{k})\cdot \mathbf{x}} f_{\mathbf{m}_j} -\frac{1}{C_{-2}}
\begin{bmatrix}
\sum f_{\mathbf{m}_j} \\ \sum f_{\mathbf{m}_j}e^{i(\xi_{\mathbf{m}_j}+\mathbf{k})\cdot\mathbf{x}_0}
\end{bmatrix}^T \\ \begin{bmatrix}
\mu & -\sum e^{i(\xi_{\mathbf{m}_j}+\mathbf{k})\cdot\mathbf{x}_0} \\ -\sum e^{-i(\xi_{\mathbf{m}_j}+\mathbf{k})\cdot\mathbf{x}_0} & \mu
\end{bmatrix}\begin{bmatrix}
\sum e^{i(\xi_{\mathbf{m}_j}+\mathbf{k})\cdot \mathbf{x}} \\ \sum e^{i(\xi_{\mathbf{m}_j}+\mathbf{k})\cdot (\mathbf{x}-\mathbf{x}_0)}
\end{bmatrix}\end{multline*}
Hence, we can show that \(P\) is the projection of \(f\) onto the eigenspace \eqref{es1} since
\[Pf(\mathbf{0})=0,\quad Pf(\mathbf{x}_0)=0\]

\emph{Case 2:} Suppose \(\mu = \left| \sum_{j=1}^\mu e^{i\xi_{\mathbf{m}_j} \cdot \mathbf{x}_0}\right| \) and \(\alpha \ne \lim_{\lambda\nearrow \lambda'} \left(g_\lambda(\mathbf{0},\mathbf{k})-\left|g_\lambda(\mathbf{x}_0,\mathbf{k})\right|\right)\). Since \(C_{-2}= 0\) and \(C_{-1}\ne 0\),
\[\begin{aligned}
Pf(\mathbf{x}) =&
\sum_{j=1}^\mu e^{i(\xi_{\mathbf{m}_j}+\mathbf{k})\cdot \mathbf{x}} f_{\mathbf{m}_j} -\frac{1}{C_{-1}}
\left( \sum f_{\mathbf{m}_j}\right)\begin{bmatrix}
1 \\ e^{i(\xi_{\mathbf{m}_1}+\mathbf{k})\cdot\mathbf{x}_0}
\end{bmatrix}^T \\& \begin{bmatrix}
\alpha-R_{\lambda'}(\mathbf{0},\mathbf{k}) & R_{\lambda'}(\mathbf{x}_0,\mathbf{k}) \\ R_{\lambda'}(-\mathbf{x}_0,\mathbf{k}) & \alpha-R_{\lambda'}(\mathbf{0},\mathbf{k})
\end{bmatrix}\begin{bmatrix}
1 \\ e^{-i(\xi_{\mathbf{m}_1}+\mathbf{k})\cdot \mathbf{x}_0}
\end{bmatrix}\left(\sum e^{i(\xi_{\mathbf{m}_j}+\mathbf{k})\cdot \mathbf{x}}\right)
\\=& \sum_{j=1}^\mu e^{i(\xi_{\mathbf{m}_j}+\mathbf{k})\cdot \mathbf{x}} f_{\mathbf{m}_j} - \left(\sum_{j=1}^\mu f_{\mathbf{m}_j}\right) \left( \sum_{j=1}^\mu e^{i(\xi_{\mathbf{m}_j}+\mathbf{k})\cdot\mathbf{x}}\right)
\end{aligned}\]
Hence, \(P\) is the projection onto \eqref{es2}.

\emph{Case 3:} Suppose \(\mu = \left| \sum_{j=1}^\mu e^{i\xi_{\mathbf{m}_j} \cdot \mathbf{x}_0}\right| \) and \(\alpha = \lim_{\lambda\nearrow \lambda'} \left(g_\lambda(\mathbf{0},\mathbf{k})-\left|g_\lambda(\mathbf{x}_0,\mathbf{k})\right|\right)\). Since \(C_{-2}=0, ~C_{-1}= 0\) and \(C_0 \ne 0\),
\[\begin{aligned}
Pf(\mathbf{x}) =&
\sum_{j=1}^\mu e^{i(\xi_{\mathbf{m}_j}+\mathbf{k})\cdot \mathbf{x}} f_{\mathbf{m}_j} -\frac{1}{C_0}
\begin{bmatrix}
\left(\overline{R_{\lambda'}}(\bullet,\mathbf{k}),f\right) \\
\left(\overline{R_{\lambda'}}(\bullet-\mathbf{x}_0,\mathbf{k}),f \right)
\end{bmatrix}^T \\& \begin{bmatrix}
\mu & -\sum_{j=1}^\mu e^{i(\xi_{\mathbf{m}_j}+\mathbf{k})\cdot\mathbf{x}_0}\\
-\sum_{j=1}^\mu e^{-i(\xi_{\mathbf{m}_j}+\mathbf{k})\cdot\mathbf{x}_0}& \mu
\end{bmatrix}\begin{bmatrix}
R_{\lambda'}(\mathbf{x},\mathbf{k}) \\
R_{\lambda'}(\mathbf{x}-\mathbf{x}_0,\mathbf{k})
\end{bmatrix}
\\&-\frac{1}{C_0}\left( \sum f_{\mathbf{m}_j}\right)\begin{bmatrix}
1 \\ e^{i(\xi_{\mathbf{m}_1}+\mathbf{k})\cdot\mathbf{x}_0}
\end{bmatrix}^T \begin{bmatrix}
-R_{\lambda'}^2(\mathbf{0}) & R_{\lambda'}^2(\mathbf{x}_0)\\ -R_{\lambda'}^2(\mathbf{0}) & R_{\lambda'}^2(\mathbf{x}_0)
\end{bmatrix} \\ & \begin{bmatrix}
1 \\ e^{-i(\xi_{\mathbf{m}_1}+\mathbf{k})\cdot \mathbf{x}_0}
\end{bmatrix}\left(\sum e^{i(\xi_{\mathbf{m}_j}+\mathbf{k})\cdot \mathbf{x}}\right)
\\=& \sum_{j=1}^\mu e^{i(\xi_{\mathbf{m}_j}+\mathbf{k})\cdot \mathbf{x}} f_{\mathbf{m}_j} - \left(\sum_{j=1}^\mu f_{\mathbf{m}_j}\right) \left( \sum_{j=1}^\mu e^{i(\xi_{\mathbf{m}_j}+\mathbf{k})\cdot\mathbf{x}}\right) \\&+
C \left(\tilde g_{2,\lambda'_-}(\bullet,\mathbf{k}),f\right) \tilde g_{2,\lambda'_-}(\mathbf{x},\mathbf{k}),\quad \text{for some }C\ne 0
\end{aligned}\]
Hence, \(P\) is the projection onto \eqref{es3}.
\end{proof}

Although the conditions for \emph{Case 2} and \emph{Case 3} seem restrictive in some sense, we can observe various cases satisfying those conditions. For example, suppose \(\mathbf{k}=(k_x,0) \in \mathcal{B}\), \(k_x>0\) and choose \(\lambda'\) as \[\lambda'=|\xi_{(0,-1)}+\mathbf{k}|=|\xi_{(-1,0)}+\mathbf{k}|\] so that \(\lambda'\) is an eigenvalue of \(-\Delta(\mathbf{k})\) of multiplicity \(\mu=2\). Then we have
\[\left| e^{i\xi_{(0,-1)}\cdot \mathbf{x}_0} + e^{i\xi_{(-1,0)}\cdot \mathbf{x}_0} \right| = \left| 2 e^{-i\frac{4\pi}{3}} \right|=2,\] which falls into either \emph{Case 2} or \emph{Case 3} of \eqref{2deltamult} depending on the value of \(\alpha\). Hence, 
\[\mult \left(\lambda',-\Delta_{\alpha,\{\mathbf{0},\mathbf{x}_0\}}(\mathbf{k})\right) = \begin{cases} 1 &\mbox{ if } \alpha \ne \lim_{\lambda\nearrow \lambda'} \left(g_\lambda(\mathbf{0},\mathbf{k})-\left|g_\lambda(\mathbf{x}_0,\mathbf{k})\right|\right) \\ 2 &\mbox{ if } \alpha = \lim_{\lambda\nearrow \lambda'} \left(g_\lambda(\mathbf{0},\mathbf{k})-\left|g_\lambda(\mathbf{x}_0,\mathbf{k})\right|\right) \end{cases}.\]

On the other hand, suppose \(\mathbf{k}=(0,k_y) \in \mathcal{B}\), \(k_y>0\) and choose \(\lambda'\) as \[\lambda'=|\xi_{(0,1)}+\mathbf{k}|=|\xi_{(-1,0)}+\mathbf{k}|.\] so that \(\mu=2\).
Then we observe 
\[\left| e^{i\xi_{(0,1)}\cdot \mathbf{x}_0} + e^{i\xi_{(-1,0)}\cdot \mathbf{x}_0} \right| = \left| e^{i\frac{4\pi}{3}}+e^{-i\frac{4\pi}{3}} \right| \ne 2,\]
which corresponds to \emph{Case 1} of \eqref{2deltamult}. Therefore, \(\mult \left(\lambda',-\Delta_{\alpha,\{\mathbf{0},\mathbf{x}_0\}}(\mathbf{k})\right) = 0\), which implies
\[\lambda'\notin \sigma \left(-\Delta_{\alpha,\{\mathbf{0},\mathbf{x}_0\}}(\mathbf{k})\right),\quad \alpha\in\mathbb{R}.\]

\begin{rmk}
If \(\mult \left(\lambda,-\Delta(\mathbf{k})\right)=1\), then \(\left|\sum_{j=1}^\mu e^{i\xi_{\mathbf{m}_j} \cdot \mathbf{x}_0}\right| =\left|e^{i\xi_{\mathbf{m}_1} \cdot \mathbf{x}_0}\right|=1 \) so this falls into either Case 2 or Case 3.
\[\begin{dcases}
\lambda \notin \sigma(-\Delta_{\alpha,\{0,\mathbf{x}_0\}}(\mathbf{k})) &\mbox{ if } \alpha \ne \lim_{\lambda\nearrow \lambda'} \left(g_\lambda(\mathbf{0},\mathbf{k})-\left|g_\lambda(\mathbf{x}_0,\mathbf{k})\right|\right)\\
\lambda \in \sigma(-\Delta_{\alpha,\{0,\mathbf{x}_0\}}(\mathbf{k})) &\mbox{ if } \alpha = \lim_{\lambda\nearrow \lambda'} \left(g_\lambda(\mathbf{0},\mathbf{k})-\left|g_\lambda(\mathbf{x}_0,\mathbf{k})\right|\right)
\end{dcases}\]
On the other hand, if \(\mult \left(\lambda,-\Delta(\mathbf{k})\right) \ge 3\), then \[\lambda \in \sigma(-\Delta_{\alpha,\{0,\mathbf{x}_0\}}(\mathbf{k})).\]
\end{rmk}

Also, we can easily show that the dispersion bands \(\lambda_1(\mathbf{k},\alpha)\le \lambda_2(\mathbf{k},\alpha)\le \cdots \) given as functions of \(\mathbf{k}\) by Theorem~\ref{thm:2deltapert} and Theorem~\ref{thm:2deltaunpert} are continuous and \(\mathcal{B}\)-periodic so they are literally "surfaces" over the Brillouin Zone.
Now consider the spectrum as a function of \(\alpha\). Note that \(\alpha=\infty\) corresponds to the free Hamiltonian \(-\Delta(\mathbf{k})\).

\begin{prop}\label{globalasymp}
Let \(\lambda_1(\mathbf{k},\alpha) \le \lambda_2(\mathbf{k},\alpha)\le \cdots \) and \(\lambda_1^\infty (\mathbf{k})\le \lambda_2^\infty (\mathbf{k})\le \cdots\) be the eigenvalues of \(-\Delta_{\alpha,\{0,\mathbf{x}_0\}}(\mathbf{k})\) and \(-\Delta(\mathbf{k})\), respectively. Then \(\alpha \mapsto \lambda_j(\mathbf{k},\alpha)\) is increasing for all \(j\) and 
\[\lambda_j(\mathbf{k},\alpha)\ge \lambda_1^\infty (\mathbf{k}) \ge 0, \quad j\ge 3,\quad \forall \alpha\in (-\infty,\infty]\]
\[\lambda_j(\mathbf{k},\alpha) \rightarrow -\infty \text{ as }\alpha \rightarrow -\infty \quad j=1,2 \]
In addition, 
\[\lim_{\alpha \rightarrow \infty} \lambda_{j}(\mathbf{k},\alpha)= \lim_{\alpha \rightarrow -\infty} \lambda_{j+2}(\mathbf{k},\alpha) =\lambda_j^\infty (\mathbf{k})\quad \text{ for all } j \ge 1\]
\end{prop}
\begin{proof}
First, \(\lambda\mapsto g_\lambda(\mathbf{0},\mathbf{k})\pm|g_\lambda(\mathbf{x}_0, \mathbf{k})|\) is continuous and increasing for both signs for \(\lambda \in \left(-\infty,\lambda_1^\infty(\mathbf{k}) \right)\). In addition, we observe that
\[\lim_{\lambda\rightarrow-\infty}g_\lambda(\mathbf{0},\mathbf{k})\pm|g_\lambda(\mathbf{x}_0, \mathbf{k})|=-\infty\]

Therefore, by \eqref{2delta+} and \eqref{2delta-}, there exist exactly two perturbed eigenvalues \(\lambda_1(\mathbf{k},\alpha)\) and \(\lambda_2(\mathbf{k},\alpha)\) in \((-\infty, \lambda_1^\infty(\mathbf{k}))\) whenever \[\alpha \le \lim_{\lambda\nearrow \lambda_1^\infty(\mathbf{k})}\left(g_\lambda(\mathbf{0},\mathbf{k})-|g_\lambda(\mathbf{x}_0, \mathbf{k})|\right).\] This also implies that 
\[\lim_{\alpha \rightarrow -\infty}\lambda_j(\mathbf{k},\alpha) \rightarrow -\infty, \quad j=1,2 \]
On the other hand, consider the other eigenvalues near an arbitrary \(\lambda'\in \sigma(-\Delta(\mathbf{k}))\). Since \(g_\lambda(\mathbf{0},\mathbf{k})\pm|g_\lambda(\mathbf{x}_0, \mathbf{k})|\) is finite for \(\lambda\notin\sigma(-\Delta(\mathbf{k}))\), we observe that 
\[\begin{dcases}\lim_{\alpha \rightarrow \infty} \lambda_{j}(\mathbf{k},\alpha) &\in \sigma(-\Delta(\mathbf{k}))\\ \lim_{\alpha \rightarrow -\infty} \lambda_{j+2}(\mathbf{k},\alpha) &\in \sigma(-\Delta(\mathbf{k})) \end{dcases} \quad \text{ for all } j \ge 1\]
Therefore, it suffices to show that those two limits agree with the same multiplicity at \(\lambda=\lambda'\) for all \(j \ge 1\). Note that as \(\lambda \rightarrow \lambda'\), \[\begin{aligned} g_\lambda(\mathbf{0},\mathbf{k})&=-\frac{\mu}{\mathrm{area}(\Gamma)} (\lambda-\lambda')^{-1} +O(1) \\
|g_\lambda(\mathbf{x}_0, \mathbf{k})|&=\frac{1}{\mathrm{area}(\Gamma)}\left| \sum_{j=1}^\mu e^{i\xi_{\mathbf{m}_j} \cdot \mathbf{x}_0}\right| |\lambda-\lambda'|^{-1} +O(1)\end{aligned}\]

Let \(\mu=\mult(\lambda',-\Delta(\mathbf{k}))\). We consider two cases as follows:

\emph{Case 1:} If \(\mu \ne \left| \sum_{j=1}^\mu e^{i\xi_{\mathbf{m}_j} \cdot \mathbf{x}_0}\right|\), then the unperturbed eigenvalue \(\lambda'\) satisfies 
\[\mult(\lambda',-\Delta_{\alpha,\{\mathbf{0},\mathbf{x}_0\}}(\mathbf{k}))=\mu-2\] according to Theorem~\ref{thm:2deltaunpert}. We also observe
\[\begin{dcases}
\lim_{\lambda\nearrow\lambda'}g_\lambda(\mathbf{0},\mathbf{k})+|g_\lambda(\mathbf{x}_0, \mathbf{k})|&=\infty\\
\lim_{\lambda\nearrow\lambda'}g_\lambda(\mathbf{0},\mathbf{k})-|g_\lambda(\mathbf{x}_0, \mathbf{k})|&=\infty\\
\lim_{\lambda\searrow\lambda'}g_\lambda(\mathbf{0},\mathbf{k})+|g_\lambda(\mathbf{x}_0, \mathbf{k})|&=\infty\\
\lim_{\lambda\searrow\lambda'}g_\lambda(\mathbf{0},\mathbf{k})-|g_\lambda(\mathbf{x}_0, \mathbf{k})|&=-\infty
\end{dcases}\]
which imply by Theorem~\ref{thm:2deltapert} that there exists exactly two perturbed eigenvalues converging to \(\lambda=\lambda'\) as \(\alpha \rightarrow \pm \infty\). So the multiplicity of \(\lambda=\lambda'\) is conserved for both cases when \(\alpha=+\infty\) and \(\alpha=-\infty\).

\emph{Case 2:} If \(\mu = \left| \sum_{j=1}^\mu e^{i\xi_{\mathbf{m}_j} \cdot \mathbf{x}_0}\right|\), then the unperturbed eigenvalue \(\lambda'\) satisfies
\[\mult(\lambda',-\Delta_{\alpha,\{\mathbf{0},\mathbf{x}_0\}}(\mathbf{k}))=\mu-1\] for \(|\alpha|\) sufficiently large according to Theorem~\ref{thm:2deltaunpert}. We also observe
\[\begin{dcases}
\lim_{\lambda\nearrow\lambda'}g_\lambda(\mathbf{0},\mathbf{k})+|g_\lambda(\mathbf{x}_0, \mathbf{k})|&=\infty\\
\lim_{\lambda\nearrow\lambda'}g_\lambda(\mathbf{0},\mathbf{k})-|g_\lambda(\mathbf{x}_0, \mathbf{k})|&=C\\
\lim_{\lambda\searrow\lambda'}g_\lambda(\mathbf{0},\mathbf{k})+|g_\lambda(\mathbf{x}_0, \mathbf{k})|&=C\\
\lim_{\lambda\searrow\lambda'}g_\lambda(\mathbf{0},\mathbf{k})-|g_\lambda(\mathbf{x}_0, \mathbf{k})|&=-\infty
\end{dcases},\quad C \text{ is finite}\]
which imply by Theorem~\ref{thm:2deltapert} that there exists exactly one perturbed eigenvalue converging to \(\lambda'\) as \(\alpha \rightarrow \pm \infty\). So the multiplicity of \(\lambda'\) is conserved for both cases when \(\alpha\rightarrow+\infty\) and \(\alpha\rightarrow-\infty\).

This concludes the proof.
\end{proof}



\subsubsection{Eigenvalues near Dirac points}
Now we investigate local behavior of dispersion bands near Dirac points, which are located at the vertices of the Brillouin Zone boundary. (See Figure~\ref{dual}.) Without loss of generality, we choose one Dirac point \[\mathbf{K}= \frac{2}{3} \mathbf{k}_1 +\frac{1}{3}\mathbf{k}_2\] out of six points and observe the conic dispersion bands generated by \(-\Delta_{\alpha,\{\mathbf{0},\mathbf{x}_0\}}(\mathbf{k})\) near \(\mathbf{k}=\mathbf{K}\).

\begin{prop}\label{diracdouble}
At Dirac point \(\mathbf{K}\), the perturbed eigenvalues of \(-\Delta_{\alpha,\{0,\mathbf{x}_0\}}(\mathbf{K})\) has multiplicity \(2\) and coincides with those of the triangular lattice operator \(-\Delta_{\alpha,\{0\}}(\mathbf{K})\), namely,
\[\sigma(-\Delta_{\alpha,\{0,\mathbf{x}_0\}}(\mathbf{K}))\setminus\sigma(-\Delta(\mathbf{K})) = \sigma(-\Delta_{\alpha,\{0\}}(\mathbf{K})\setminus\sigma(-\Delta(\mathbf{K}))\] and for all \(\lambda' \in \sigma(-\Delta_{\alpha,\{0,\mathbf{x}_0\}}(\mathbf{K}))\setminus\sigma(-\Delta(\mathbf{K}))\),
\[\mult (\lambda',-\Delta_{\alpha,\{0,\mathbf{x}_0\}}(\mathbf{K}))=2 \]

\end{prop}

\begin{proof}
We use the symmetry of the honeycomb structure and dual lattice as in Section 2.4 of \cite{honeycomb}.Let \(\tilde{R}:\mathbb{R}^2 \rightarrow \mathbb{R}^2\) be a \(\frac{2\pi}{3}\)-rotation around a Dirac Point \(-\mathbf{K}\). Then we see
\begin{equation}\label{tildeR}\begin{aligned}
\tilde{R}\xi_\mathbf{m}&=\tilde{R}\xi_{(m_1,m_2)}=\xi_{(-m_1+m_2-1,-m_1-1)}\\
\tilde{R}^2\xi_\mathbf{m}&=\tilde{R}^2\xi_{(m_1,m_2)}=\xi_{(-m_2-1,m_1-m_2)}\\
\tilde{R}^3\xi_\mathbf{m}&=\mathrm{Id}~\xi_{(m_1,m_2)}=\xi_{(m_1,m_2)}
\end{aligned}\end{equation}
Here we abuse the notation and write \(\tilde{R}\xi_\mathbf{m} = \xi_{\tilde{R}\mathbf{m}}\) so that
\[\begin{aligned}
\tilde{R}\mathbf{m}&=\tilde{R}(m_1,m_2)=(-m_1+m_2-1,-m_1-1)\\
\tilde{R}^2\mathbf{m}&=\tilde{R}^2(m_1,m_2)=(-m_2-1,m_1-m_2)\\
\tilde{R}^3\mathbf{m}&=\mathrm{Id}~(m_1,m_2)=(m_1,m_2)
\end{aligned}\]
So we can decompose \(\mathbb{Z}^2 \) into three disjoint subsets \(\mathcal{S}, \tilde{R}\mathcal{S},\) and \(\tilde{R}^2\mathcal{S}\) where \[\mathbb{Z}^2=\mathcal{S}\cup \tilde{R}\mathcal{S} \cup \tilde{R}^2\mathcal{S}\]

For instance, \(\{(0,0), (-1,-1), (-1,0)\}\) is an orbit of \(\tilde{R}\). So we choose exactly one of them, say \((0,0)\), as an element of \(\mathcal{S}\). (See Definition 2.4 of \cite{honeycomb}.) Since \(\mathbf{x}_0=\frac{2}{3}(\mathbf{v}_1+\mathbf{v}_2)\), we obtain
\[\begin{aligned}
\xi_\mathbf{m} \cdot \mathbf{x}_0 &= \frac{2\pi}{3}(m_1+m_2)\\
\xi_{\tilde{R}m} \cdot \mathbf{x}_0 &= \frac{2\pi}{3}(-2m_1+m_2-2)=\frac{2\pi}{3}(m_1+m_2)-(2\pi m_1+\frac{4\pi}{3})\\
\xi_{\tilde{R}^2m} \cdot \mathbf{x}_0 &= \frac{2\pi}{3}(m_1-2m_2-1)=\frac{2\pi}{3}(m_1+m_2)-(2\pi m_2+\frac{2\pi}{3})\\
\end{aligned}\]
Therefore,
\[e^{i\xi_\mathbf{m}\cdot \mathbf{x}_0}+e^{i\xi_{\tilde{R}m}\cdot \mathbf{x}_0}+e^{i\xi_{\tilde{R}^2m}\cdot \mathbf{x}_0} = e^{i\frac{2\pi}{3}(m_1+m_2)}\left( 1+e^{i\frac{2\pi}{3}}+e^{i\frac{4\pi}{3}}\right)=0\]
In addition, note that \[|\xi_\mathbf{m}+\mathbf{K}|=|\xi_{\tilde{R}m}+\mathbf{K}|=|\xi_{\tilde{R}^2m}+\mathbf{K}|\]
We can conclude that 
\[\begin{aligned}
g_{\lambda'}(\mathbf{x}_0,\mathbf{K})
&=\sum_{\mathbf{m}\in\mathbb{Z}^2}{\frac{e^{i\xi_\mathbf{m}\cdot \mathbf{x}_0}}{|\xi_\mathbf{m}^2+\mathbf{K}|^2-\lambda'}}\\
&=\sum_{\mathbf{m}\in \mathcal{S}}{\left[ \frac{e^{i\xi_\mathbf{m}\cdot \mathbf{x}_0}}{|\xi_\mathbf{m}^2+\mathbf{K}|^2-\lambda'}
+\frac{e^{i\xi_{\tilde{R}\mathbf{m}}\cdot \mathbf{x}_0}}{|\xi_{\tilde{R}\mathbf{m}}^2+\mathbf{K}|^2-\lambda'}
+\frac{e^{i\xi_{\tilde{R}^2\mathbf{m}}\cdot \mathbf{x}_0}}{|\xi_{\tilde{R}^2\mathbf{m}}^2+\mathbf{K}|^2-\lambda'}\right]}\\
&=\sum_{\mathbf{m}\in\mathcal{S}}{\frac{e^{i\xi_\mathbf{m}\cdot \mathbf{x}_0}+e^{i\xi_{\tilde{R}\mathbf{m}}\cdot \mathbf{x}_0}+e^{i\xi_{\tilde{R}^2\mathbf{m}}\cdot \mathbf{x}_0}}{|\xi_\mathbf{m}^2+\mathbf{K}|^2-\lambda'}}\\
&=0
\end{aligned}\]

Hence, \eqref{2delta+} and \eqref{2delta-} become two identical formulae and 
\[\mult (\lambda',-\Delta_{\alpha,\{0,\mathbf{x}_0\}}(\mathbf{K}))=2\quad \forall \lambda' \in \sigma(-\Delta_{\alpha,\{0,\mathbf{x}_0\}}(\mathbf{K}))\setminus\sigma(-\Delta(\mathbf{K})) .\]
Moreover, \eqref{2delta+} and \eqref{2delta-} coincide with \eqref{1delta} of the triangular lattice case.
\[\sigma(-\Delta_{\alpha,\{0,\mathbf{x}_0\}}(\mathbf{K}))\setminus\sigma(-\Delta(\mathbf{K})) = \sigma(-\Delta_{\alpha,\{0\}}(\mathbf{K})\setminus\sigma(-\Delta(\mathbf{K}))\]
\end{proof}

In addition, those eigenvalues of multiplicity 2 given in the previous proposition are the conic points on pairs of dispersion bands.
\begin{lem}\label{coniclem}
For each \(\lambda' \in \sigma(-\Delta_{\alpha,\{0,\mathbf{x}_0\}}(\mathbf{K}))\setminus\sigma(-\Delta(\mathbf{K}))\), there exist a pair of dispersion bands \(\mathbf{k} \mapsto \lambda_-(\mathbf{k})\) and \(\mathbf{k} \mapsto \lambda_+(\mathbf{k})\) of the operator \(-\Delta_{\alpha\{0,\mathbf{x}_0\}}(\mathbf{k})\) such that \[\lambda'=\lambda_-(\mathbf{K}) = \lambda_+(\mathbf{K}). \] In addition, \(\lambda_+\) and \(\lambda_-\) meet conically at \((\mathbf{K}. \lambda')\) with the directional derivatives independent of the direction \(\mathbf{u}\in \mathbb{S}^1\),
\[\nabla_\mathbf{u}\lambda_\pm (\mathbf{K})= \pm c(\lambda')\] where \(c(\lambda') > 0\) is defined by
\begin{equation}\label{c}
c(\lambda')=\dfrac{4\pi}{a} \dfrac{\left| \sum_{\mathbf{m}\in\mathbb{Z}^2} \dfrac{m_1 e^{i\xi_\mathbf{m}\cdot\mathbf{x}_0} }{(|\xi_\mathbf{m}+\mathbf{K}|^2-\lambda')^2} \right|} {\sum_{\mathbf{m}\in\mathbb{Z}^2} \dfrac{1}{(|\xi_\mathbf{m}+\mathbf{K}|^2-\lambda')^2}}, \qquad \mathbf{m}=(m_1,m_2) \end{equation}
\end{lem}


\begin{proof}
Suppose \(\lambda' \in \sigma(-\Delta_{\alpha,\{0,\mathbf{x}_0\}}(\mathbf{K}))\setminus\sigma(-\Delta(\mathbf{K}))\). By Proposition~\ref{diracdouble}, \(\lambda'\) is an eigenvalue of \(-\Delta_{\alpha,\{0,\mathbf{x}_0\}}(\mathbf{K})\) of multiplicity 2. So we can choose \(j\) such that \(\lambda'=\lambda_j(\mathbf{K},\alpha)=\lambda_{j+1}(\mathbf{K},\alpha)\). Since \(\alpha\) is fixed, let
\[\begin{aligned}
\lambda_+(\mathbf{k})&=\lambda_{j+1}(\mathbf{k},\alpha)\\
\lambda_-(\mathbf{k})&=\lambda_{j}(\mathbf{k},\alpha)
\end{aligned}\]

Now we prove the conic behavior of \(\mathbf{k}\mapsto\lambda_+(\mathbf{k})\) near \(\lambda=\lambda'\).
Consider a small perturbation \(\delta\mathbf{k}\) and the corresponding \(\delta\lambda\): 
\begin{align*}\delta \mathbf{k}&=|\delta \mathbf{k}|\mathbf{u} \in \mathbb{R}^2, \quad |\delta \mathbf{k}| \ll 1, \quad \mathbf{u}\in \mathbb{S}^1 \\
\delta\lambda &=\lambda_+(\mathbf{K}+\delta\mathbf{k})- \lambda_+(\mathbf{K})\end{align*}
Then both \((\mathbf{K},\lambda')\) and \((\mathbf{K}+\delta\mathbf{k},\lambda'+\delta\lambda)\) solve \eqref{2delta+}:
\begin{equation}\label{conic+1}
\alpha-g_{\lambda'}(\mathbf{0},\mathbf{K})+|g_{\lambda'}(\mathbf{x}_0,\mathbf{K})|=0
\end{equation}
and 
\begin{equation}\label{conic+2}\begin{aligned}
\alpha-g_{\lambda'+\delta\lambda}(\mathbf{0},\mathbf{K}+\delta\mathbf{k})+|g_{\lambda'+\delta\lambda}(\mathbf{x}_0,\mathbf{K}+\delta\mathbf{k})|=0
\end{aligned}\end{equation}

Subtracting \eqref{conic+1} from \eqref{conic+2}, we obtain
\begin{equation}\label{conic+3}
\sum_{m\in\mathbb{Z}^2}{\frac{2(\xi_\mathbf{m}+\mathbf{K})\cdot\delta\mathbf{k}-\delta\lambda}{(|\xi_\mathbf{m}+\mathbf{K}|^2-\lambda')^2}}
+\left| \sum_{m\in\mathbb{Z}^2}{\frac{e^{i \xi_\mathbf{m}\cdot \mathbf{x}_0}\left(2(\xi_\mathbf{m}+\mathbf{K})\cdot\delta\mathbf{k}-\delta\lambda\right)}{(|\xi_\mathbf{m}+\mathbf{K}|^2-\lambda')^2}}\right|=o(|\delta\mathbf{k}|)
\end{equation}

In addition, we observe that 
\[\sum_{m\in\mathbb{Z}^2}{\frac{e^{i \xi_\mathbf{m}\cdot \mathbf{x}_0}}{(|\xi_\mathbf{m}+\mathbf{K}|^2-\lambda')^2}}=0\] 
and 
\[\sum_{m\in\mathbb{Z}^2}{\frac{\xi_\mathbf{m}+\mathbf{K}}{(|\xi_\mathbf{m}+\mathbf{K}|^2-\lambda')^2}}=\mathbf{0}\]
due to the symmetry property as in Proposition~\ref{diracdouble}. So we can simplify \eqref{conic+3} as
\[-\sum_{m\in\mathbb{Z}^2}{\frac{\delta\lambda}{(|\xi_\mathbf{m}+\mathbf{K}|^2-\lambda')^2}}+
\left| \sum_{m\in\mathbb{Z}^2}{2\frac{e^{i \xi_\mathbf{m}\cdot \mathbf{x}_0}(\xi_\mathbf{m}+\mathbf{K})\cdot\delta\mathbf{k}}{(|\xi_\mathbf{m}+\mathbf{K}|^2-\lambda')^2}}\right|=o(|\delta\mathbf{k}|)\]

Hence, as \(|\delta\mathbf{k}| \rightarrow 0\), we obtain the directional derivative of the upper dispersion band,
\[\nabla_\mathbf{u}\lambda_+ (\mathbf{K})= + \dfrac{\left| \sum_{\mathbf{m}\in\mathbb{Z}^2} \dfrac{2e^{i\xi_\mathbf{m}\cdot\mathbf{x}_0}(\xi_\mathbf{m}+\mathbf{K})}{(|\xi_\mathbf{m}+\mathbf{K}|^2-\lambda')^2}\cdot \mathbf{u} \right|}{\sum_{\mathbf{m}\in\mathbb{Z}^2} \dfrac{1}{(|\xi_\mathbf{m}+\mathbf{K}|^2-\lambda')^2}}.\]

Define \(\mathbf{c}: \mathbb{R} \rightarrow \mathbb{C}^2\) by \[\mathbf{c}(\lambda')=\dfrac{ \sum_{\mathbf{m}\in\mathbb{Z}^2} \dfrac{2e^{i\xi_\mathbf{m}\cdot\mathbf{x}_0}(\xi_\mathbf{m}+\mathbf{K})}{(|\xi_\mathbf{m}+\mathbf{K}|^2-\lambda')^2} }{\sum_{\mathbf{m}\in\mathbb{Z}^2} \dfrac{1}{(|\xi_\mathbf{m}+\mathbf{K}|^2-\lambda')^2}}.\] so that 
\[\nabla_\mathbf{u}\lambda_+ (\mathbf{K})=|\mathbf{c}(\lambda')\cdot\mathbf{u}|.\]

Then we observe that
\[e^{i \frac{\pi}{3}} \mathbf{c}(\lambda')\cdot \frac{\mathbf{v}_1}{a} = \mathbf{c}(\lambda')\cdot \frac{\mathbf{v}_2}{a} \]
which implies for any \(\mathbf{u}\in\mathbb{S}^1\),
\[\begin{aligned}
|\mathbf{c}(\lambda')\cdot \mathbf{u}|&=\left| \mathbf{c}(\lambda')\cdot \frac{\mathbf{v}_1}{a} \right| 
\\&= \dfrac{4\pi}{a} \dfrac{\left| \sum_{\mathbf{m}\in\mathbb{Z}^2} \dfrac{m_1 e^{i\xi_\mathbf{m}\cdot\mathbf{x}_0} }{(|\xi_\mathbf{m}+\mathbf{K}|^2-\lambda')^2} \right|} {\sum_{\mathbf{m}\in\mathbb{Z}^2} \dfrac{1}{(|\xi_\mathbf{m}+\mathbf{K}|^2-\lambda')^2}}
\\&=c(\lambda')
\end{aligned}\]

This concludes that the directional derivative is independent of the direction \(\mathbf{u}\in\mathbb{S}^1\). We can show that \(\nabla_\mathbf{u}\lambda_- (\mathbf{K})=-c(\lambda')\) by similar considerations.
\end{proof}

\begin{rmk}
The parameters \(\boldsymbol\alpha=(\alpha,\alpha)\) and the location of the point \(\mathbf{x}_0=\frac{2}{3}(\mathbf{v}_1+\mathbf{v}_2)\) in the fundamental domain \(\Gamma\) both play crucial roles in existence of conic singularities on dispersion bands since perturbed eigenvalues \(\lambda'\in \sigma(-\Delta_{\alpha,Y} (\mathbf{k}))\setminus\sigma(-\Delta(\mathbf{k}))\) of multiplicity 2 can only be obtained when the matrix \(\Gamma_{\alpha,\{\mathbf{0},\mathbf{x}_0\}}(\lambda',\mathbf{k})\) defined by \eqref{gamma} is a \(2 \times 2\) diagonal matrix with identical diagonal entries. In other words, \(\mathbf{x}_0\) determines if \(\Gamma_{\alpha,\{\mathbf{0},\mathbf{x}_0\}}(\lambda',\mathbf{k})\) is diagonal and the identical parameters \(\boldsymbol\alpha=(\alpha,\alpha)\) make those diagonal entries equal to each other so that the two dispersion bands meet at Dirac Point.
\end{rmk}

Now we introduce one of the main theorems about the conic singularities of a honeycomb lattice point scatterer. We observe two pairs of conic surfaces above and below the third dispersion bands as follows:
\begin{thm}\label{conicthm}
For all \(\alpha \in \mathbb{R}\) and \(\mathbf{k}\in\mathcal{B}\), \[\lambda_1 (\mathbf{k},\alpha)=\lambda_2 (\mathbf{k},\alpha)\] and \[\lambda_4 (\mathbf{k},\alpha)=\lambda_5 (\mathbf{k},\alpha).\] In addition, as \(\mathbf{k}\rightarrow\mathbf{K}\), 
\[\begin{aligned}
\lambda_1(\mathbf{k},\alpha)-\lambda_1(\mathbf{K},\alpha)&=- c(\lambda_1(\mathbf{K},\alpha)) \left| \mathbf{k}-\mathbf{K}\right|+o(|\mathbf{k}-\mathbf{K}|)
\\
\lambda_{2}(\mathbf{k},\alpha)-\lambda_{2}(\mathbf{K},\alpha)&= c(\lambda_{1}(\mathbf{K},\alpha)) \left| \mathbf{k}-\mathbf{K}\right|+o(|\mathbf{k}-\mathbf{K}|)
\\
\lambda_{4}(\mathbf{k},\alpha)-\lambda_{4}(\mathbf{K},\alpha)&=- c(\lambda_{4}(\mathbf{K},\alpha)) \left| \mathbf{k}-\mathbf{K}\right|+o(|\mathbf{k}-\mathbf{K}|)
\\
\lambda_{5}(\mathbf{k},\alpha)-\lambda_{5}(\mathbf{K},\alpha)&=c(\lambda_{4}(\mathbf{K},\alpha)) \left| \mathbf{k}-\mathbf{K}\right|+o(|\mathbf{k}-\mathbf{K}|)
\end{aligned}\]

\end{thm}
\begin{proof}
By Proposition~\ref{thm:2deltaunpert}, \[\lambda_1^\infty(\mathbf{K}) = |\xi_{(0,0)}+\mathbf{K}|=|\xi_{(-1,-1)}+\mathbf{K}|=|\xi_{(-1,0)}+\mathbf{K}|\] is an eigenvalue of \(-\Delta_{\alpha,\{0,\mathbf{x}_0\}}(\mathbf{K})\) of multiplicity 1. Then \(\lambda_3(\mathbf{K},\alpha)=\lambda_1^\infty(\mathbf{K})\) since we obtain \(\lambda_1(\mathbf{K},\alpha)=\lambda_2(\mathbf{K},\alpha)\) and \(\lambda_4(\mathbf{K},\alpha)=\lambda_5(\mathbf{K},\alpha)\) by solving \eqref{1delta} on \((-\infty,\lambda_1^\infty(\mathbf{K}))\) and on \((\lambda_1^\infty(\mathbf{K}),\lambda_2^\infty(\mathbf{K}))\), respectively.
Then we obtain the desired conclusion by applying Lemma~\ref{coniclem} to \(\lambda_1(\mathbf{K},\alpha)=\lambda_2(\mathbf{K},\alpha)\) and \(\lambda_4(\mathbf{K},\alpha)=\lambda_5(\mathbf{K},\alpha)\).
\end{proof}

\begin{rmk}
In addition to the two pairs of conic singularities described in Theorem~\ref{conicthm}, we observe by Lemma~\ref{coniclem} that there exist infinitely many pairs of dispersion bands sharing the same property. More precisely, for each perturbed eigenvalue \(\lambda' \in \sigma(-\Delta_{\alpha,\{0,\mathbf{x}_0\}}(\mathbf{K}))\setminus\sigma(-\Delta(\mathbf{K}))\) at a Dirac point \(\mathbf{K}\), there exist two dispersion bands passing through \(\lambda'\) forming a pair of conic singularities.
On the other hand, by Proposition~\ref{thm:2deltaunpert}, we observe a different kind of \((3n-2)\)-fold conic singularities at \((\mathbf{K},\lambda')\) where \(\lambda' \in \sigma(-\Delta(\mathbf{K}))\) and \(n\) defined by \eqref{n} is greater than 1.
For instance, the \(j\)-th dispersion bands with \(j=9,\cdots, 12\) form a 4-fold conic singularity at Dirac point.
See Figure~\ref{floqaloc} for both kinds of conic surfaces at Dirac point.
\end{rmk}


Due to the global asymptotic behavior of dispersion bands described in Theorem~\ref{globalasymp}, the conic points also approach other surfaces as \(\alpha \rightarrow \pm \infty \). In particular, as \(\alpha\rightarrow\infty\), \[\lambda_1(\mathbf{K},\alpha)=\lambda_2(\mathbf{K},\alpha) \nearrow |\mathbf{K}^2|\]
where \(\lambda_3(\mathbf{K},\alpha)=\lambda_1^\infty(\mathbf{K}) = |\mathbf{K}^2|\) is independent of \(\alpha\).
On the other hand, as \(\alpha\rightarrow-\infty\),
\[\lambda_1(\mathbf{K},\alpha)=\lambda_2(\mathbf{K},\alpha) \rightarrow -\infty \]
\[\lambda_4(\mathbf{K},\alpha)=\lambda_5(\mathbf{K},\alpha) \searrow |\mathbf{K}^2| .\]
See Supplemental materials for continuous transition of dispersion bands where \(\alpha\) is moving from \(100\) to \(-100\).

\begin{figure}
\centering
\begin{subfigure}[b]{0.47\textwidth}
\centering
\includegraphics[width=\textwidth]{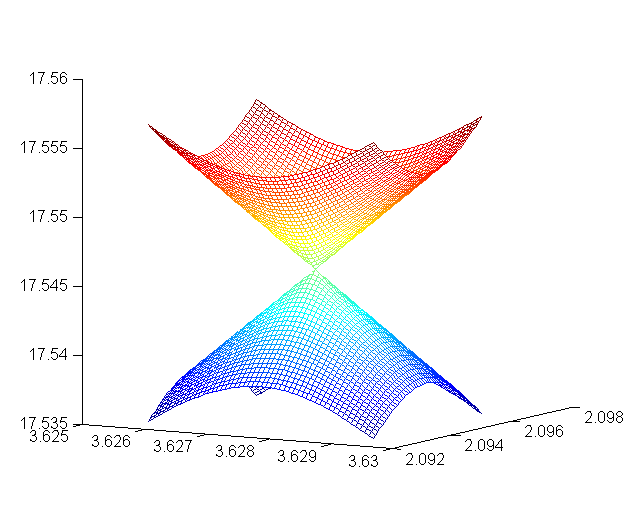} 
\caption{Triangular lattice, \(j=2,3\)}
\label{floq1a0loc23}
\end{subfigure}
\begin{subfigure}[b]{0.47\textwidth}
\centering
\includegraphics[width=\textwidth]{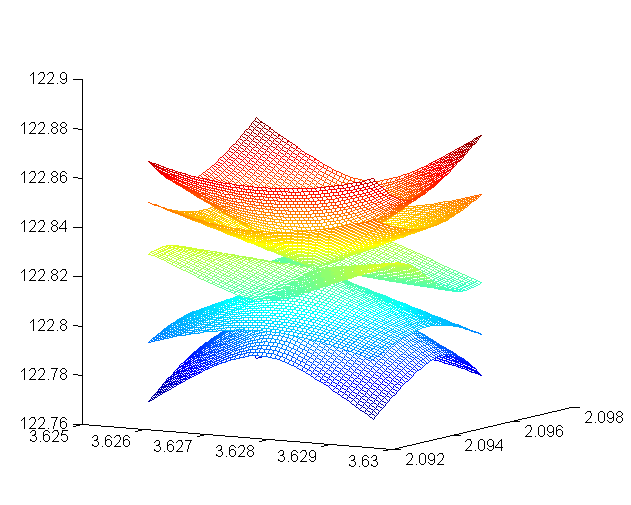} 
\caption{Triangular lattice, \(j=8,\cdots,12\)}
\label{floq1am1loc}
\end{subfigure}
\begin{subfigure}[b]{0.47\textwidth}
\centering
\includegraphics[width=\textwidth]{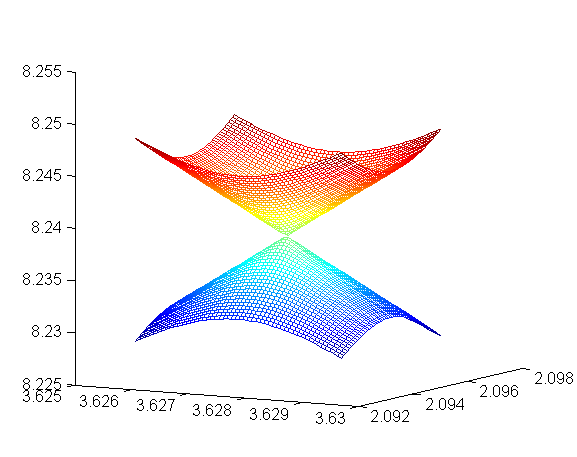} 
\caption{Honeycomb lattice, \(j=1,2\)}
\label{floq2a0loc12}
\end{subfigure}
\begin{subfigure}[b]{0.47\textwidth}
\centering
\includegraphics[width=\textwidth]{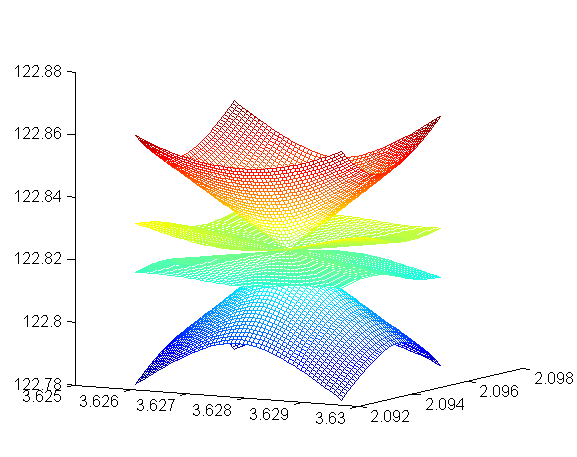} 
\caption{Honeycomb lattice, \(j=9,\cdots,12\)}
\label{floq2a0loc912}
\end{subfigure}

\caption{Local view of the \(j\)-th dispersion band generated by a triangular lattice point scatterer \(-\Delta_{0,\{\mathbf{0}\}}(\mathbf{k})\), and by a honeycomb lattice point scatterer\(-\Delta_{0,\{\mathbf{0},\mathbf{x}_0\}}(\mathbf{k})\) near Dirac Point \(\mathbf{K}\). See Supplemental materials for the first five dispersion bands with \(\alpha\) moving from \(100\) to \(-100\).
}\label{floqaloc}
\end{figure}

\subsubsection{Spectral properties of the full Hamiltonian}
By Proposition~\ref{prop:fiber}, the full Hamiltonian of the honeycomb lattice point scatterer \(-\Delta_{\boldsymbol\beta,H}\), \(\boldsymbol\beta=(\alpha,\alpha,\cdots)\) generates a spectrum given by the union of \(\sigma\left(-\Delta_{\alpha,\{0,\mathbf{x}_0\}}(\mathbf{k})\right)\) for all \(\mathbf{k} \in \mathcal{B}\). Furthermore, \(\sigma\left(-\Delta_{\boldsymbol\beta,H}\right)\) formulates a band structure due to the continuity of each dispersion band. More precisely, we obtain a union of infinitely many intervals corresponding to all dispersion bands:
\begin{equation}\label{minmaxinterval}\sigma\left(-\Delta_{\boldsymbol\beta,H}\right) = \bigcup_{j\ge 1} \left[\min_{\mathbf{k}\in\mathcal{B}} \lambda_j(\mathbf{k},\alpha),~\max_{\mathbf{k}\in\mathcal{B}} \lambda_j(\mathbf{k},\alpha) \right].\end{equation}

Although a generic dispersion band \(\mathbf{k}\mapsto \lambda_j (\mathbf{k},\alpha)\) attains its maximum and minimum at various points \(\mathbf{k}\in\mathcal{B}\) that depends on \(\alpha\), we observe a simple result at least for the lowest level.
\begin{prop}\label{prop:l10}
For any \(\alpha \in (-\infty,\infty]\), the lowest eigenvalue of the first band occurs at \(\mathbf{k}=\mathbf{0}\), namely,
\[\min_{\mathbf{k}\in\mathcal{B}}\lambda_1 (\mathbf{k},\alpha) = \lambda_1 (\mathbf{0},\alpha).\] 
\end{prop}
\begin{proof}
This can be proved easily when \(\alpha=\infty\) since \[\lambda_1 (\mathbf{k},\infty)=\lambda_1^\infty (\mathbf{k})=|\mathbf{k}|^2.\]

Now suppose \(\alpha \in \mathbb{R}\). Then \(\lambda_1 (\mathbf{k},\alpha)\) is obtained by solving \eqref{2delta-} on \((-\infty,\lambda_1^\infty (\mathbf{k}))\) where the RHS of the equation is continuous and increasing as a function of \(\lambda\). Since \(\lambda_1 (\mathbf{k},\alpha) < \lambda_1^\infty(\mathbf{k})\), it suffices to show that 
\[\alpha \ge g_{\lambda_1 (\mathbf{0},\alpha)}(\mathbf{0},\mathbf{k})+\left|g_{\lambda_1 (\mathbf{0},\alpha)}(\mathbf{x}_0,\mathbf{k})\right|\]
or equivalently,
\[g_{\lambda_1 (\mathbf{0},\alpha)}(\mathbf{0},\mathbf{0})+\left|g_{\lambda_1 (\mathbf{0},\alpha)}(\mathbf{x}_0,\mathbf{0})\right| \ge g_{\lambda_1 (\mathbf{0},\alpha)}(\mathbf{0},\mathbf{k})+\left|g_{\lambda_1 (\mathbf{0},\alpha)}(\mathbf{x}_0,\mathbf{k})\right|.\]
By the Poisson summation formula (Lemma~III.4.4) in \cite{solvable}, for \(\Im \sqrt{\lambda}>0\) and \(\mathbf{x}\in\Gamma\),
\[g_\lambda (\mathbf{x},\mathbf{k}) = \begin{dcases} \sum_{\mathbf{v}\in\Lambda} G_\lambda (\mathbf{x}+\mathbf{v}) e^{-i\mathbf{k}\cdot\mathbf{v}} &\mbox{ if } \mathbf{x} \in \Gamma \setminus \{\mathbf{0}\} \\
\sum_{\substack{\mathbf{v}\in\Lambda \\ \mathbf{v}\ne \mathbf{0}}} G_\lambda (\mathbf{v}) e^{-i \mathbf{k}\cdot\mathbf{v}} +\frac{1}{2\pi}\ln\left(\frac{\sqrt{\lambda}}{i}\right) &\mbox{ if } \mathbf{x} =\mathbf{0} 
\end{dcases}\]
Since \(\lambda_1 (\mathbf{0},\alpha) <\lambda_1^\infty(\mathbf{0})=0\), we have \(G_{\lambda_1 (\mathbf{0},\alpha)}(\mathbf{x})>0\) for \(\mathbf{x}\in \mathbb{R}^2\setminus \{\mathbf{0}\}\). So we obtain 
\[\begin{aligned}
&g_{\lambda_1 (\mathbf{0},\alpha)}(\mathbf{0},\mathbf{k})+\left|g_{\lambda_1 (\mathbf{0},\alpha)}(\mathbf{x}_0,\mathbf{k})\right|\\
&=\sum_{\substack{\mathbf{v}\in\Lambda \\ \mathbf{v}\ne \mathbf{0}}} G_{\lambda_1 (\mathbf{0},\alpha)} (\mathbf{v}) e^{-i \mathbf{k}\cdot\mathbf{v}} +\frac{1}{2\pi}\ln\left(\frac{\sqrt{\lambda_1 (\mathbf{0},\alpha)}}{i}\right) +\left|\sum_{\mathbf{v}\in\Lambda} G_{\lambda_1 (\mathbf{0},\alpha)} (\mathbf{x}_0+\mathbf{v}) e^{-i\mathbf{k}\cdot\mathbf{v}}\right|\\
&\le \sum_{\substack{\mathbf{v}\in\Lambda \\ \mathbf{v}\ne \mathbf{0}}} G_{\lambda_1 (\mathbf{0},\alpha)} (\mathbf{v}) +\frac{1}{2\pi}\ln\left(\frac{\sqrt{\lambda_1 (\mathbf{0},\alpha)}}{i}\right) +\left|\sum_{\mathbf{v}\in\Lambda} G_{\lambda_1 (\mathbf{0},\alpha)} (\mathbf{x}_0+\mathbf{v}) \right|\\
&=g_{\lambda_1 (\mathbf{0},\alpha)}(\mathbf{0},\mathbf{0})+\left|g_{\lambda_1 (\mathbf{0},\alpha)}(\mathbf{x}_0,\mathbf{0})\right|
\end{aligned}\]
which implies \[\min_{\mathbf{k}\in\mathcal{B}} \lambda_1 (\mathbf{k},\alpha) = \lambda_1 (\mathbf{0},\alpha)\]
\end{proof}

Recall that the full Hamiltonian of a triangular lattice point scatterer has a spectrum consisting of at most two intervals by Proposition~\ref{prop:floq1spec}. We now observe an analogous result for honeycomb lattice point scatterers. See Figure~\ref{floq2spec} for the graph of spectral bands \(\sigma(-\Delta_{\boldsymbol\beta,H})\) versus \(\alpha\in\mathbb{R}\) where \(\boldsymbol\beta=(\alpha,\alpha,\cdots)\).


\begin{thm} \label{band}
For \(\alpha\in \mathbb{R}\) and \(\boldsymbol\beta=(\alpha,\alpha,\cdots)\), the spectrum of \(-\Delta_{\boldsymbol\beta,H}\) consists of at most three disjoint intervals, namely,
\[\sigma(-\Delta_{\boldsymbol\beta,H})= I_1 \cup I_2 \cup I_3\]
where
\begin{equation}\label{I1}
I_1 = \left[ \lambda_1 (\mathbf{0},\alpha), ~\max_{\mathbf{k}\in\mathcal{B}}\lambda_2 (\mathbf{k},\alpha) \right]
\end{equation}
\begin{equation}\label{I2}
I_2 = \left[ \min_{\mathbf{k}\in\mathcal{B}}\lambda_3 (\mathbf{k},\alpha), ~\max_{\mathbf{k}\in\mathcal{B}}\lambda_3 (\mathbf{k},\alpha) \right]
\end{equation}
\begin{equation}\label{I3}
I_3 = \left[ \min_{\mathbf{k}\in\mathcal{B}}\lambda_4 (\mathbf{k},\alpha), ~\infty \right)
\end{equation}
\end{thm}
\begin{proof}
\eqref{I2} is trivial due to the continuity of each dispersion band \(\mathbf{k} \mapsto \lambda_j(\mathbf{k},\alpha)\). \eqref{I1} can also be easily shown since we have \[\min_{\mathbf{k}\in\mathcal{B}}\lambda_1 (\mathbf{k},\alpha) = \lambda_1 (\mathbf{0},\alpha)\] and \(\lambda_1 (\mathbf{K},\alpha)=\lambda_2 (\mathbf{K},\alpha)\) by Proposition~\ref{prop:l10} and Proposition~\ref{diracdouble}, respectively.

Now we prove all spectral bands except the three lowest levels overlap with adjacent band spectra to produce a single interval \(I_3\).
Fix \(\alpha\in\mathbb{R}\) and assume that there exists \(\tilde\lambda \notin \sigma(-\Delta_{\alpha,\{0,\mathbf{x}_0\}}(\mathbf{K}))\) such that \[\tilde\lambda>\min_{\mathbf{k}\in\mathcal{B}} \lambda_4 (\mathbf{k},\alpha)\]
Then we can choose \(j' \ge 4\) such that
\[\lambda_{j'}(\mathbf{k},\alpha)<\tilde\lambda<\lambda_{j'+1}(\mathbf{k},\alpha), \quad \mathbf{k}\in \mathcal{B}\]
Also, it was proved by S.Albeverio \cite{solvable2} that 
\begin{equation}\label{lambdabound}\lambda_{j-2}^\infty (\mathbf{k}) \le \lambda_j(\mathbf{k},\alpha) \le \lambda_{j}^\infty (\mathbf{k}),\quad \mathbf{k}\in \mathcal{B},~ j\ge 3 \end{equation} 
where \(\lambda_j^\infty(\mathbf{k})\) is defined as in Theorem~\ref{globalasymp}.
So we have for all \(\mathbf{k}\in\mathcal{B}\),
\begin{equation}\label{solvable21}
\lambda_{j'-2}^\infty(\mathbf{k})<\tilde\lambda<\lambda_{j'+1}^\infty(\mathbf{k}).\end{equation}
On the other hand, we know by the symmetry property of the dual lattice \(\Lambda^*\) that 
\begin{equation}\label{dualsymm0}\lambda_{6j+2}^\infty(\mathbf{0})=\lambda_{6j+3}^\infty(\mathbf{0})=\cdots=\lambda_{6j+7}^\infty(\mathbf{0}), \quad j \ge 0\end{equation} and 
\begin{equation}\label{dualsymmK}\lambda_{3j+1}^\infty(\mathbf{K})=\lambda_{3j+2}^\infty(\mathbf{K})=\lambda_{3j+3}^\infty(\mathbf{K}), \quad j \ge 0.\end{equation}
By \eqref{solvable21} and \eqref{dualsymm0}, we can choose \(j''\ge1\) such that
\[\lambda_{6j''+1}^\infty(\mathbf{0})<\tilde\lambda<\lambda_{6j''+2}^\infty(\mathbf{0}).\]
Then by the assumption on \(\tilde{\lambda}\) and continuity of \(\mathbf{k}\mapsto\lambda_j^\infty(\mathbf{k})\), we obtain \[\lambda_{6j''+1}^\infty(\mathbf{K})<\tilde\lambda<\lambda_{6j''+2}^\infty(\mathbf{K})\] which contradicts \eqref{dualsymmK}. Hence, \[\bigcup_{j\ge4} \{\lambda_j(\mathbf{k},\alpha)~|~\mathbf{k}\in\mathcal{B}\}= \left[\min_{\mathbf{k}\in\mathcal{B}}\lambda_4\left(\mathbf{k},\alpha\right),\infty \right)\]
\end{proof}

\begin{figure}
\centering

\begin{subfigure}[b]{0.47\textwidth}
\centering
\includegraphics[width=\textwidth]{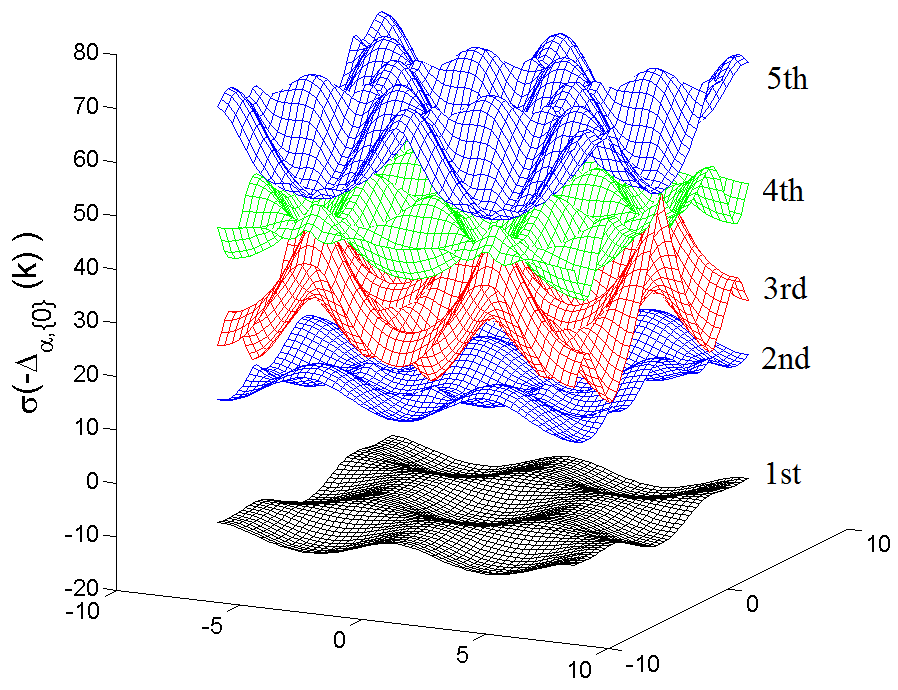} 
\caption{Triangular lattice}
\label{floq1aglobal}
\end{subfigure}
\begin{subfigure}[b]{0.47\textwidth}
\centering
\includegraphics[width=\textwidth]{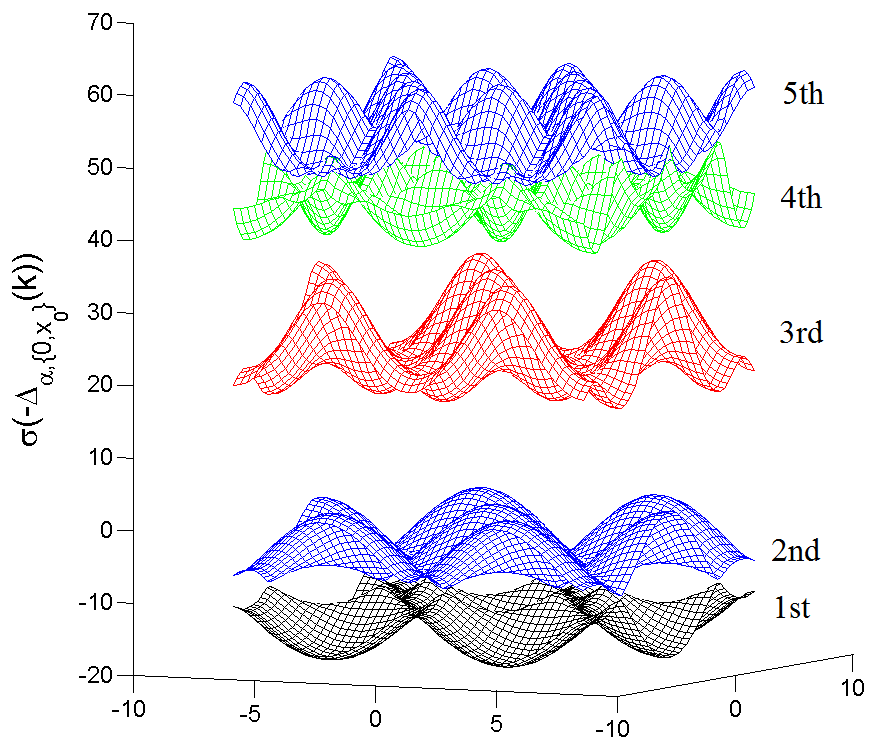} 
\caption{Honeycomb lattice}
\label{floq2aglobal}
\end{subfigure}

\caption{Front views of the first five dispersion bands generated by a triangular lattice point scatterer \(-\Delta_{\alpha,\{\mathbf{0}\}}(\mathbf{k})\) and a honeycomb lattice point scatterer \(-\Delta_{(\alpha,\alpha),\{\mathbf{0},\mathbf{x}_0\}}(\mathbf{k})\) with \(\alpha\) fixed at \(-0.07\). See Supplemental materials for the first five dispersion bands with \(\alpha\) moving from \(100\) to \(-100\).
}\label{floqaglobal}
\end{figure}

\begin{figure}
\centering

\begin{subfigure}[b]{0.47\textwidth}
\centering
\includegraphics[width=\textwidth]{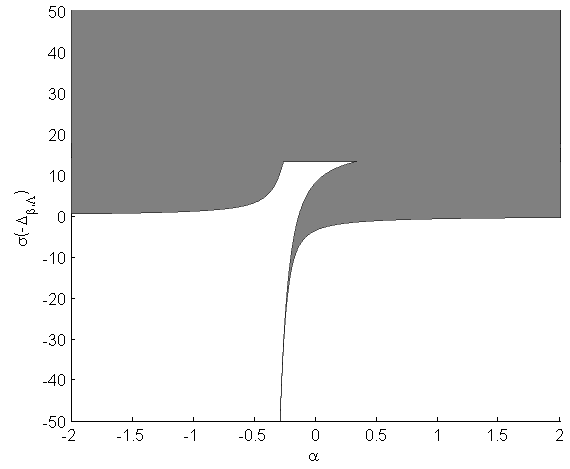} 
\caption{Triangular lattice}
\label{floq1spec}
\end{subfigure}
\begin{subfigure}[b]{0.47\textwidth}
\centering
\includegraphics[width=\textwidth]{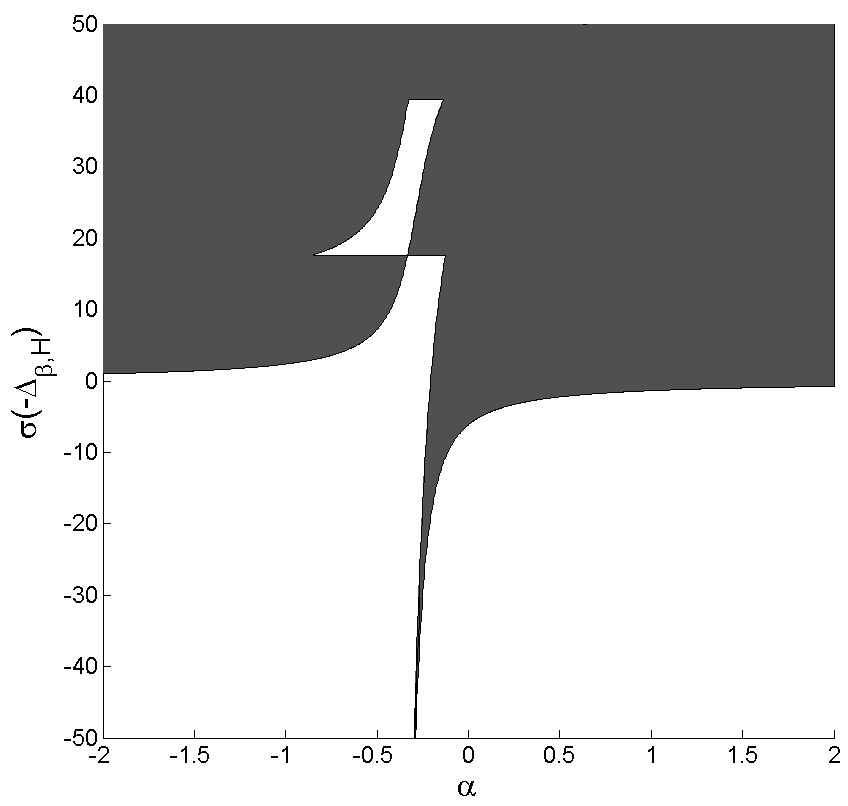} 
\caption{Honeycomb lattice}
\label{floq2spec}
\end{subfigure}

\caption{Spectrum of the triangular lattice point scatterer \(-\Delta_{\boldsymbol\beta,\Lambda}\) and honeycomb lattice point scatterer \(-\Delta_{\boldsymbol\beta,H}\) with identical parameters \(\boldsymbol\beta=(\alpha,\alpha,\cdots)\) where \(\alpha\in\mathbb{R}\). The vertical section of the shaded region at each \(\alpha\) represents \(\sigma(-\Delta_{\boldsymbol\beta,\Lambda})\) and \(\sigma(-\Delta_{\boldsymbol\beta,H})\).
}\label{floqspec}
\end{figure}

\appendix
\section{Introduction to Periodic Point Scatterers}\label{sec:introtoptsctr}
In this section, we briefly introduce some important results already known for periodic point scatterers. First, Floquet theory for smooth periodic potentials will be discussed. See also \cite{reedsimon4} for a rigorous mathematical formulation and \cite{solidstate} for a discussion in a more physical context. Then we summarize the corresponding theory of periodic point scatterers in \(\mathbb{R}^2\). See Part III of \cite{solvable} for more details.

\subsection{Floquet Theory} \label{sec:floquet}
Suppose \(V:\mathbb{R}^2\rightarrow\mathbb{R}\) is smooth and periodic. Then we can expand \(V\) in a Fourier series:
\[V(\mathbf{x})=\sum_{\mathbf{m}\in\mathbb{Z}^2} V_\mathbf{m} e^{i\xi_\mathbf{m}\cdot\mathbf{x}}\]
with the Fourier coefficients \[V_\mathbf{m}=\frac{1}{\mathrm{area}(\Gamma)} \int_{\Gamma} V(\mathbf{x})e^{-i \xi_\mathbf{m}\cdot \mathbf{x}} d\mathbf{x}.\]

Consider the Fourier transform of \(P=-\Delta+V\). \(\Delta\) denotes the Laplacian.
\begin{equation}\label{fourierP}
\hat{P}\hat{f}(\xi)=\mathcal{F}P\mathcal{F}^{-1}\hat{f}(\xi) = |\xi|^2 \hat{f}(\xi) + \frac{1}{2\pi} \int_{\mathbb{R}^2}\hat{V}(\xi-\eta)\hat{f}(\eta) d\eta,\quad \xi\in\mathbb{R}^2\end{equation}
By the Fourier inversion formula, we formally obtain 
\begin{equation}\label{fourierV}
\hat{V}(\xi)=2\pi \sum_{\mathbf{m}\in\mathbb{Z}^2} V_\mathbf{m}\delta(\xi-\xi_\mathbf{m}) \end{equation}
since
\[\begin{aligned}
V(\mathbf{x})&=\sum_{\mathbf{m}\in\mathbb{Z}^2} V_\mathbf{m} e^{i\xi_\mathbf{m}\cdot\mathbf{x}}\\
&= \sum_{\mathbf{m}\in\mathbb{Z}^2} V_\mathbf{m} \int_{\mathbb{R}^2} \delta(\xi-\xi_\mathbf{m})e^{i\xi\cdot x} d\xi \\
&= \frac{1}{2\pi}\int_{\mathbb{R}^2} \hat{V}(\xi)e^{i\xi\cdot\mathbf{x}}
\end{aligned}\]
Combining \eqref{fourierP} and \eqref{fourierV}, we obtain
\begin{equation}\label{fourierP2}
\hat{P}\hat{f}(\xi)=|\xi|^2 \hat{f}(\xi) +\sum_{\mathbf{m}\in\mathbb{Z}^2} V_\mathbf{m} \hat{f}(\xi-\xi_\mathbf{m}) \end{equation}
Also, we can always decompose \(\xi\in\mathbb{R}^2\) into a lattice point \(\xi_\mathbf{m}\in\Lambda^*\) and a remainder \(\mathbf{k}\in\mathcal{B}\) so that \begin{equation}\label{ximk}\xi=\xi_\mathbf{m}+\mathbf{k}.\end{equation}

Hence, for each \(\mathbf{k}\in\mathcal{B}\), the eigenvalue problem \((P-\lambda)f=0,~\lambda\in \mathbb{R}\) becomes an algebraic problem for \(\{\hat{f}(\xi_\mathbf{m}+\mathbf{k})~|~\mathbf{m}\in\mathbb{Z}^2\}\) as follows:
\begin{equation}\label{centraleq}
(|\xi_\mathbf{m}+\mathbf{k}|^2-\lambda) \hat{f}(\xi_\mathbf{m}+\mathbf{k}) +\sum_{\mathbf{m}'\in\mathbb{Z}^2} V_{\mathbf{m}'} \hat{f}(\xi_{\mathbf{m}-\mathbf{m}'}+\mathbf{k})=0, \quad \mathbf{m}\in\mathbb{Z}^2
\end{equation}
which is called the ``central equation'' \cite{solidstate} for Floquet-Bloch theory. 

In order to study the spectrum of the full Hamiltonian \(-\Delta+V\) in detail using the central equation, we consider a decomposition operator \(\mathcal{U}: L^2(\mathbb{R}^2) \rightarrow \int_{\mathcal{B}}^{\oplus}d\mathbf{k} ~\ell^2(\Lambda^*) \) defined as
\[\mathcal{U}\hat{f} (\mathbf{k},\xi_\mathbf{m})=\hat{f} (\xi_\mathbf{m}+\mathbf{k}), \quad \hat{f}\in L^2(\mathbb{R}^2)\]
which corresponds to the decomposition of a vector in \eqref{ximk}.

Then we obtain
\begin{equation}\label{upu}
\left(\mathcal{U}^{-1}\hat{P}\mathcal{U} \hat{f}\right)(\mathbf{k},\xi_\mathbf{m})=\left(\hat{P}(\mathbf{k})\hat{f}(\mathbf{k},\bullet)\right)(\xi_\mathbf{m})\end{equation}
with \(\hat{P}(\mathbf{k}): \ell^2 (\Lambda^*)\rightarrow \ell^2 (\Lambda^*)\) defined for each \(\mathbf{k}\in\mathcal{B}\) by
\begin{equation}\label{centraleq2}\hat{P}(\mathbf{k}) g(\xi_\mathbf{m}) = |\xi_\mathbf{m}+\mathbf{k}|^2 g(\xi_\mathbf{m})+\sum_{\mathbf{m}'\in \mathbb{Z}^2} V_{\mathbf{m}'} g(\xi_{\mathbf{m}}-\xi_{\mathbf{m}'}) \end{equation}

Note that \eqref{upu} can also be written
\begin{equation}\label{upu2}
\mathcal{U}^{-1}\hat{P}\mathcal{U} = \int_\mathcal{B}^\oplus d\mathbf{k} \hat{P}(\mathbf{k}).\end{equation}

In addition, we obtain a similar result to \eqref{centraleq2} in the \(\mathbf{x}\)-space. Define \(P(\mathbf{k}):L_\mathbf{k}^2(\Gamma)\rightarrow L_\mathbf{k}^2(\Gamma)\) as the inverse Fourier transform of the operator \(\hat{P}(\mathbf{k})\), namely, \[P(\mathbf{k})=\mathcal{F}^{-1}\hat{P}(\mathbf{k})\mathcal{F}=-\Delta(\mathbf{k})+V\]

In particular, if \(V \equiv 0\), we observe from \eqref{centraleq2} that the spectrum of the free decomposed Hamiltonian \(-\hat{\Delta}(\mathbf{k})\) is discrete. More precisely,
\[\sigma(-\Delta(\mathbf{k}))=\sigma_d (-\Delta(\mathbf{k}))=\left\{ |\xi_\mathbf{m}+\mathbf{k}|^2 ~|~ \mathbf{m}\in\mathbb{Z}^2 \right\}\]

Similarly, the non-zero potentials also produce discrete spectra determined by \eqref{centraleq2}. Then we observe 
\[\sigma(P)=\sigma(\hat{P})=\bigcup_{\mathbf{k}\in\mathcal{B}}\sigma (\hat{P}(\mathbf{k})) =\bigcup_{\mathbf{k}\in\mathcal{B}}\sigma (P(\mathbf{k})).\]
Hence, the spectrum of \(P\) appears as a union of bands where the \(j\)-th band consists of the \(j\)-th eigenvalue of \(P(\mathbf{k}),~\mathbf{k}\in\mathcal{B}\) if the dependence in \(\mathbf{k}\) is continuous.

\subsection{Point scatterers on periodic structures}
We define the point scatterers on a periodic structure using self-adjoint extension theory and renormalization process. For simplicity's sake, we will focus on formulating periodic point scatterers on a honeycomb lattice \[H=\Lambda+Y = \{\tilde{\mathbf{y}}_1,\tilde{\mathbf{y}}_2,\tilde{\mathbf{y}}_3,\cdots\}\] with \(\Lambda=\mathbb{Z}\mathbf{v}_1\oplus \mathbb{Z}\mathbf{v}_2,\) and \(Y=\{\mathbf{y}_1,\mathbf{y}_2\}=\{\mathbf{0},\mathbf{x}_0\}\) as defined in Section~\ref{2lat}. However, the whole process in this section can be also applied to the case of a lattice different from \(\Lambda\) and a finite set different from \(Y\) in the fundamental domain of the lattice. In particular, if we pick \(Y= \{\mathbf{0}\}\) with the same \(\Lambda\), then we can formulate periodic point scatterers on a triangular lattice.

We first investigate a self-adjoint operator \[-\Delta_{\boldsymbol\beta,H},\quad\boldsymbol\beta=(\beta_1,\beta_2,\cdots)\] on \(L^2(\mathbb{R}^2)\) for infinitely many point scatterers on \(H\) where the parameter \(\beta_j \in (-\infty,\infty]\) describes the strength of the point scatterer at \(\tilde{\mathbf{y}}_j\in H\). For those parameters, we will follow Albeverio's notation in \cite{solvable} so that the point scatterer at \(\tilde{\mathbf{y}}_j\) gets strong when \(|\beta_j| \ll \infty\) while it vanishes and acts as the Laplacian near \(\tilde{\mathbf{y}}_j\) when \(\beta_j=\infty\).

More precisely, we formulate a self-adjoint operator for finitely many point scatterers in \(\mathbb{R}^2\) using Von Neumann's theory of self-adjoint extension \cite{reedsimon2} and Krein's formula \cite{krein}. Then we consider \(-\Delta_{\boldsymbol\beta,H}\) as the limit of such an operator in the norm resolvent sense. 

See Chapter III.4 of \cite{solvable} for the complete formulation. Then the resolvent \[ (-\Delta_{\boldsymbol\beta,H}-\lambda)^{-1}: L^2(\mathbb{R}^2)\rightarrow L^2(\mathbb{R}^2)\] reads for \(f\in L^2(\mathbb{R}^2)\), \(\mathbf{x}\in\mathbb{R}^2\) and \(\lambda \notin \sigma(-\Delta_{\boldsymbol\beta,H})\cup [0,\infty)\),
\begin{multline}\label{infmany}
(-\Delta_{\boldsymbol\beta,H}-\lambda)^{-1}f (\mathbf{x}) = \\G_\lambda f (\mathbf{x}) + \sum_{j,j'=1}^\infty [\Gamma_{\boldsymbol\beta,H}(\lambda)^{-1}]_{jj'} \left(G_{\lambda}(\bullet,\tilde{\mathbf{y}}_{j'}),f \right) G_{\lambda}(\mathbf{x},\tilde{\mathbf{y}}_j)
\end{multline}
where
\(G_\lambda = (-\Delta-\lambda)^{-1}\) is the integral kernel of the free resolvent in \(L^2(\mathbb{R}^2)\) explicitly defined by
\begin{equation}\label{kernelr2} 
G_\lambda(\mathbf{x},\mathbf{x}')=\frac{i}{4} H_0^{(1)}(\sqrt{\lambda}|\mathbf{x}-\mathbf{x}'|), \quad \Im \sqrt{\lambda} \ge 0,~ \mathbf{x}\ne \mathbf{x}'\end{equation}
where \(H_0^{(1)}\) is the Hankel function of first kind of order zero and
\(\Gamma_{\boldsymbol\beta,H}(\lambda)\) is a closed operator in \(\ell^2(H)\) given by 
\[[\Gamma_{\boldsymbol\beta,H}(\lambda)]_{jj'}=\begin{dcases}
\beta_j +\frac{1}{2\pi}\left(\ln\frac{\sqrt{\lambda}}{2i}\right) ,&\mbox{ if } j=j',~ j,j'\in \mathbb{N}\\
-G_\lambda(\tilde{\mathbf{y}}_{j'},\tilde{\mathbf{y}}_{j}), &\mbox{ if } j\ne j',~j,j'\in \mathbb{N}
\end{dcases}\]
See Theorem 4.1 in \cite{solvable} for the complete proof. Note that \(\beta_j\) in this section corresponds to \(\alpha_j+\frac{\gamma}{2\pi}\) in \cite{solvable}.

In addition, if the parameters of point scatterers on each translated lattice \(\Lambda+\mathbf{y}_j\), \(j=1,2\) are identical, then we call them \emph{periodic point scatterers} and the operator can be decomposed into ``fibers" \cite{reedsimon4}. More precisely, suppose that the parameters at periodic points are identical, namely,
\begin{multline}\label{beta}\alpha _j = \beta _{j'}\in (-\infty,\infty],\quad j=1,2, ~j'\in\mathbb{N}, \\ \text{if and only if } ~ \tilde{\mathbf{y}}_{j'} = \mathbf{y}_j +\mathbf{v}_\mathbf{m} ~\text{ for some }~ \mathbf{m}\in\mathbb{Z}^2.\end{multline}

Then we can follow the idea of the Floquet theory with a periodic potential \(V\) defined by
\[V(\mathbf{x})=\sum_{j=1}^2\sum_{\mathbf{v}\in\Lambda} c_j \delta(\mathbf{x}-\mathbf{v}-\mathbf{y}_j),\quad c_j \in \mathbb{R},\]
which provides a simplified model of a crystal consisting of two kinds of atoms whose nuclei are located on the honeycomb lattice \(H\). 

Then the corresponding Fourier coefficients become
\begin{equation}\label{Vm}
V_\mathbf{m}=-\frac{1}{\mathrm{area}(\Gamma)} \sum_{j=1}^2 c_j e^{-i \xi_\mathbf{m} \cdot \mathbf{y}_j},\quad \mathbf{m}\in \mathbb{Z}^2. \end{equation}

Hence, we obtain the central equation for a formal operator \(P=-\Delta+V(\mathbf{x})\) from \eqref{centraleq2}. For \(g \in \ell_0^2(\Lambda^*)=\left\{h\in \ell^2 (\Lambda^*)~|~ \mathrm{supp}~h \text{ is finite}\right\}\),
\begin{equation}\label{centraleq3}
\hat{P}(\mathbf{k}) g(\xi_\mathbf{m}) = |\xi_\mathbf{m}+\mathbf{k}|^2 g(\xi_\mathbf{m})- \frac{1}{\mathrm{area}(\Gamma)} \sum_{j=1}^2 \left[ c_j e^{-i\xi_\mathbf{m}\cdot \mathbf{y}_j} \sum_{\mathbf{m}'\in \mathbb{Z}^2} e^{i \xi_{\mathbf{m}'} \cdot \mathbf{y}_j}g(\xi_{\mathbf{m}'})\right]
\end{equation}

However, this is obviously not a well-defined self-adjoint operator in \(\ell^2(\Lambda^*)\) since the infinite sum on the right-hand side does not always converge for \(g\in \ell^2(\Lambda^*)\). So we propose a renormalization by the momentum cutoff \(|\xi_\mathbf{m}+\mathbf{k}|<r\). Then we consider the coefficients \(c_j\) as a function of \(r\) so that they approach zero as \(r \rightarrow \infty\). In other words, this renormalization process will make the summation in \eqref{centraleq3} converge by weakening the delta potentials in some sense so we can have the Fourier transform of the point scatterer as a well-defined self-adjoint operator in \(\ell^2(\Lambda^*)\). So we obtain the resolvent of \(-\hat{\Delta}_{\boldsymbol\alpha,Y} (\mathbf{k})\), the Fourier transform of decomposed operator \(-\Delta_{\boldsymbol\alpha,Y} (\mathbf{k})\), as in the following proposition. See Theorem III.1.4.1 and Theorem III.4.3 in \cite{solvable} for the proof. Note that \(\alpha_j\) in this proposition corresponds to \(\alpha_j+\frac{\gamma}{2\pi}\) in \cite{solvable}.

\begin{prop}
Let \(\hat{P}^r (\mathbf{k})\) be a self-adjoint operator such that 
\[(\hat{P}^r (\mathbf{k})g) (\xi_\mathbf{m}) = |\xi_\mathbf{m}+\mathbf{k}|^2 g(\xi_\mathbf{m})- \frac{1}{\mathrm{area}(\Gamma)} \sum_{j=1}^2 \left[ c_j( r) (\phi_{\mathbf{y}_j}^ r (\mathbf{k}),g)\phi_{\mathbf{y}_j}^ r (\mathbf{k}) \right]\]
where \((\bullet,\bullet)\) is the inner product in \(\ell^2(\Lambda^*)\) and \(\phi_{\mathbf{y}_j}^ r (\mathbf{k})\) is the function
\[\phi_{\mathbf{y}_j}^ r(\mathbf{k},\xi_\mathbf{m})=\chi_ r (\xi_\mathbf{m}+\mathbf{k}) e^{-i(\xi_\mathbf{m}+\mathbf{k})\cdot\mathbf{y}_j}\]
with domain 
\[D(\hat{P}^ r (\mathbf{k})) = D(-\hat{\Delta}(\mathbf{k})) = \left\{ g\in \ell^2(\Lambda^*) \middle| \sum_{\mathbf{m} \in \mathbb{Z}^2} |\xi_\mathbf{m}+\mathbf{k} |^4 g(\xi_\mathbf{m})^2 < \infty\right\}.\]

If
\[ c_j( r)=\left( \alpha_j + \frac{\ln r-\ln 2}{2\pi} \right)^{-1}, \quad \alpha_j\in (-\infty,\infty], ~ r > 0\]
then for all \(\mathbf{k} \in \mathcal{B}\), \(\hat{P}^ r (\mathbf{k})\) converges in norm resolvent sense as \( r \rightarrow \infty\) to a self-adjoint operator \(-\hat{\Delta}_{\boldsymbol\alpha,Y} (\mathbf{k})\) with resolvent
\begin{multline}\label{fourierfiber}(-\hat{\Delta}_{\boldsymbol\alpha,Y} (\mathbf{k})-\lambda)^{-1}
= G_\lambda(\mathbf{k}) + \frac{1}{\mathrm{area}(\Gamma)} \sum_{j=1}^2 [\Gamma_{\boldsymbol\alpha,Y}(\lambda,\mathbf{k})^{-1}]_{jj'} (F_{\lambda,\mathbf{y}_{j}}(\mathbf{k}),\bullet) F_{\lambda,\mathbf{y}_{j'}}(\mathbf{k}),
\\
\lambda \notin \left\{|\xi_\mathbf{m}+\mathbf{k}|^2 ~| ~\xi_\mathbf{m} \in \Lambda^* \right\}, ~\boldsymbol\alpha=(\alpha_1, \alpha_2)
\end{multline}
where 
\begin{equation}\label{glambda}
g_\lambda(\mathbf{x},\mathbf{k})= \begin{dcases}
\frac{1}{\mathrm{area}(\Gamma)}\sum_{m\in\mathbb{Z}^2}{\frac{e^{i(\xi_\mathbf{m}+\mathbf{k})\cdot \mathbf{x}}}{|\xi_\mathbf{m}^2+\mathbf{k}|^2-\lambda}} &\mbox{if } \mathbf{x} \notin \Lambda\\
\frac{1}{\mathrm{area}(\Gamma)} \lim_{ r \rightarrow \infty}{\left[\sum_{\substack{m\in\mathbb{Z}^2\\|\xi_\mathbf{m}+\mathbf{k}| \le r}}{\frac{1}{|\xi_\mathbf{m}+\mathbf{k}|^2-\lambda}}-\frac{2\pi}{\mathrm{area}(\mathcal{B})} \ln r\right]} &\mbox{if } \mathbf{x} \in \Lambda\\
\end{dcases}\end{equation}

\begin{equation}
\label{gamma}
\left[ \Gamma_{\boldsymbol\alpha,Y} (\lambda,\mathbf{k}) \right]_{jj'}= \alpha_j \delta_{jj'}-g_\lambda (\mathbf{y}_{j'}-\mathbf{y}_{j},\mathbf{k}),\quad j,j'=1,2.
\end{equation}
and 
\begin{equation}
\label{Flambda}
F_{\lambda, \mathbf{y}_j}(\mathbf{k},\xi_\mathbf{m})=\frac{e^{-i(\xi_\mathbf{m}+\mathbf{k})\cdot\mathbf{y}_j}}{|\xi_\mathbf{m}+\mathbf{k}|^2-\lambda}
\end{equation}
and \(G_\lambda(\mathbf{k}):\ell^2(\Lambda^*)\rightarrow \ell^2(\Lambda^*)\) is the multiplication operator 
\[(G_\lambda(\mathbf{k}) g)(\xi_\mathbf{m})=(|\xi_\mathbf{m}+\mathbf{k}|^2-\lambda)^{-1} g(\xi_\mathbf{m})\] 
\end{prop}

Finally, we apply the inverse Fourier transform to \(-\hat{\Delta}_{\boldsymbol\alpha,Y} (\mathbf{k})\) and obtain the resolvent of the decomposed operator \( -\Delta_{\boldsymbol\alpha,Y}(\mathbf{k})\) on \(L_\mathbf{k}^2(\Gamma)\). See Theorem III.1.4.3 and Theorem III.4.3 in \cite{solvable} for more details.
\begin{prop}\label{ndelta}
Let \(\boldsymbol\alpha =(\alpha_1,\alpha_2) \in (-\infty,\infty]^2\). For \(\lambda \notin \sigma(-\Delta (\mathbf{k}))\), the resolvent \(\left( -\Delta_{\boldsymbol\alpha,Y}(\mathbf{k})-\lambda \right)^{-1}:L_\mathbf{k}^2(\Gamma) \rightarrow L_\mathbf{k}^2(\Gamma)\) is defined by
\begin{multline}
\label{infdelta}
\left( -\Delta_{\boldsymbol\alpha,Y}(\mathbf{k})-\lambda \right)^{-1}f (\mathbf{x})=  \left(-\Delta (\mathbf{k})-\lambda \right)^{-1}f(\mathbf{x}) \\+ \frac{1}{\mathrm{area}(\mathcal{B})}\sum_{j,j'=1}^{2}{\left[\Gamma_{\boldsymbol\alpha,Y} (\lambda,\mathbf{k})^{-1}\right]_{jj'} \left(\overline{g_\lambda \left(\bullet-\mathbf{y}_{j},\mathbf{k} \right)},f\right) g_\lambda \left(\mathbf{x}-\mathbf{y}_{j'}, \mathbf{k}\right)}
\end{multline}
where \(g_\lambda(x,\mathbf{k})\) and \(\Gamma_{\boldsymbol\alpha,Y} (\lambda,\mathbf{k})\) are defined in \eqref{glambda} and \eqref{gamma}, respectively.
\end{prop}

If the infinite sequence of parameters \(\boldsymbol\beta\) satisfies the periodicity condition \eqref{beta}, we can find unitary equivalence between \(-\hat{\Delta}_{\boldsymbol\beta,H}\), the Fourier transform of the full Hamiltonian in \eqref{infmany}, and the direct integral over \(\mathcal{B}\) of the decomposed operator in \eqref{fourierfiber}. Note that this result corresponds to the Floquet theory for smooth potentials in Section~\ref{sec:floquet}. See Theorem III.1.4.3 and Theorem III.4.5 in \cite{solvable}.

\begin{prop}\label{prop:fiber} Suppose \(\boldsymbol\alpha\) and \(\boldsymbol\beta\) satisfy \eqref{beta}. Then \(-\hat{\Delta}_{\boldsymbol\beta,H}\) is unitarily equivalent to \(\int_{\mathcal{B}}^\otimes -\hat{\Delta}_{\boldsymbol\alpha,Y} (\mathbf{k}) d\mathbf{k}\) so 
\begin{equation}\label{unitequiv}
\bigcup_{\mathbf{k}\in\mathcal{B}}{\sigma(-\Delta_{\boldsymbol\alpha,Y}(\mathbf{k}))}=\bigcup_{\mathbf{k}\in\mathcal{B}}{\sigma(-\hat{\Delta}_{\boldsymbol\alpha,Y}(\mathbf{k}))}=\sigma(-\hat{\Delta}_{\boldsymbol\beta,H})=\sigma(-\Delta_{\boldsymbol\beta,H})\end{equation} 
\end{prop}
Note that the decomposed self-adjoint operator \(-\Delta_{\boldsymbol\alpha,Y}(\mathbf{k})\) has a purely discrete spectrum. Hence, as in the periodic potential case the eigenvalues of \(\Delta_{\boldsymbol\beta,H}\) also form a band structure if the dependence in \(\mathbf{k}\) is continuous.

On the other hand, Proposition~\ref{ndelta} provides additional information about the domain of \(-\Delta_{\boldsymbol\alpha,Y}(\mathbf{k})\). The proof is similar to those of Theorem I.1.1.3 and Theorem II.1.1.3  in \cite{solvable}.
\begin{prop}\label{prop:ndeltadomain}
Let \(\psi \in D(-\Delta_{\boldsymbol\alpha,Y}(\mathbf{k}))\). Then for any \[\lambda \notin \sigma(-\Delta(\mathbf{k}))\cup \sigma(-\Delta_{\boldsymbol\alpha,Y}(\mathbf{k})),\] there exists a unique \(\phi_\lambda \in H_\mathbf{k}^2(\Gamma)\) such that 
\begin{equation}\label{ndeltadomain}
\psi(\mathbf{x})=\phi_\lambda(\mathbf{x})+\frac{1}{\mathrm{area}(\mathcal{B})} \sum_{j,j'=1}^{2}{\left[\Gamma_{\boldsymbol\alpha,Y} (\lambda,\mathbf{k})^{-1}\right]_{jj'} \phi_\lambda(\mathbf{y}_j ) g_\lambda \left(\mathbf{x}-\mathbf{y}_{j'}, \mathbf{k}\right)}.\end{equation}
In addition,
\[(-\Delta_{\boldsymbol\alpha,Y}(\mathbf{k})-\lambda)^{-1}\psi=(-\Delta(\mathbf{k}))-\lambda)^{-1}\phi_\lambda.\]
\end{prop}



\FloatBarrier

\section*{Acknowledgements}
\thispagestyle{empty}
The author is greatly indebted to Maciej Zworski for suggesting the topic as well as providing guidance throughout the research.
The author was supported by the Samsung Scholarship.

\nocite{*}
\bibliographystyle{siam}
\bibliography{conic_bib}

\end{document}